\documentclass[12pt]{article}
\usepackage[latin9]{inputenc}
\usepackage{geometry}
\usepackage{color}
\usepackage{amsmath}
\usepackage{amsthm}
\usepackage{amssymb}
\usepackage{graphicx}
\usepackage{setspace}
\usepackage[authoryear]{natbib}
\usepackage[unicode=true,pdfusetitle,
 bookmarks=true,bookmarksnumbered=true,bookmarksopen=true,bookmarksopenlevel=1,
 breaklinks=true,pdfborder={0 0 1},backref=page,colorlinks=true]
 {hyperref}
\hypersetup{
 pdfborderstyle=,pdfborderstyle={},urlcolor=blue,linkcolor=red,citecolor=blue}

\makeatletter
\theoremstyle{plain}
\newtheorem{assumption}{\protect\assumptionname}
\theoremstyle{definition}
\newtheorem{defn}{\protect\definitionname}
\theoremstyle{definition}
 \newtheorem{example}{\protect\examplename}
\theoremstyle{plain}
\newtheorem{prop}{\protect\propositionname}
\theoremstyle{plain}
\newtheorem{thm}{\protect\theoremname}
\theoremstyle{remark}
\newtheorem{rem}{\protect\remarkname}
\theoremstyle{plain}
\newtheorem{lem}{\protect\lemmaname}

\makeatother

\providecommand{\assumptionname}{Assumption}
\providecommand{\definitionname}{Definition}
\providecommand{\examplename}{Example}
\providecommand{\lemmaname}{Lemma}
\providecommand{\propositionname}{Proposition}
\providecommand{\remarkname}{Remark}
\providecommand{\theoremname}{Theorem}

\begin{document}
\global\long\def\bl{\underline{b}}%

\global\long\def\bh{\bar{b}}%

\title{Auction Design with Data-Driven Misspecifications\thanks{We thank Paul Klemperer and Martin Weidner for helpful comments. The
contents of this paper were presented in a preliminary shape as the
Hurwicz lecture at the 2019 Conference on Economic Design. Jehiel
thanks the European Research Council for funding (grant no 742816).}}
\author{Philippe Jehiel and Konrad Mierendorff\thanks{Jehiel: Paris School of Economics and University College London (email:
jehiel@enpc.fr); Mierendorff: Department of Economics, University
College London (email: k.mierendorff@ucl.ac.uk)}}
\maketitle
\begin{abstract}
We consider auction environments in which at the time of the auction
bidders observe signals about their ex-post value. We introduce a
model of novice bidders who do not know know the joint distribution
of signals and instead build a statistical model relating others'
bids to their own ex post value from the data sets accessible from
past similar auctions. Crucially, we assume that only ex post values
and bids are accessible while signals observed by bidders in past
auctions remain private. We consider steady-states in such environments,
and importantly we allow for correlation in the signal distribution.
We first observe that data-driven bidders may behave suboptimally
in classical auctions such as the second-price or first-price auctions
whenever there are correlations. Allowing for a mix of rational (or
experienced) and data-driven (novice) bidders results in inefficiencies
in such auctions, and we show the inefficiency extends to all auction-like
mechanisms in which bidders are restricted to submit one-dimensional
(real-valued) bids.\smallskip{}

\noindent \textbf{Keywords}: Belief Formation, Auctions, Efficiency\\
\textbf{JEL Classification Numbers:} D44, D82, D90 
\end{abstract}

\section{Introduction}

In the standard model of auctions, bidders hold private information
about the value of the object for sale. They commonly know how information
is distributed across bidders, and every bidder $i$ correctly understands
how bidder $j$ chooses his bid as a function of his private information.
Every bidder best-responds to this correct understanding given the
rules of the auction. The resulting strategy profile is a Bayes Nash
Equilibrium.

A classic rationale for the correct understanding assumption in Bayes
Nash Equilibrium is based on learning \citep[see for example][]{Dekel2004}.
If similar auctions are played many times (by different subjects in
the roles of the various bidders), by looking at previous bids as
well as what the corresponding bidders knew (at the time of the auction),
one can recover the mapping from information to bids for past observations.
If a steady state has been reached, this mapping will correctly describe
behavior in the current auction, thereby supporting the rational expectation
formulation.

But, access to previous bidders' information is not so natural in
practice. If the private information held by previous bidders is not
disclosed, it is less clear how a novice bidder (who would only be
exposed to data from past similar auctions played by others) would
manage to have a correct understanding of other bidders' strategies.
A different modeling is required to describe the behavior of novice
bidders in such a case.

To make progress in this unexplored direction, we consider one-object
auctions in which at the time of the auction, bidder $i$'s private
information is a noisy signal about his (own) ex post value for the
good assumed to take one of finitely many realizations. We assume
that after the auction is completed, what is publicly disclosed is
the profile of bids as well as the ex post values for the various
bidders, but not the signals observed by the bidders at the time of
the auction.\footnote{One could argue that in a number of cases, only the information for
the winner is available, but we leave the study of this for future
research and focus on the non-availability of the private information
held by the bidders at the time of the auction (see also the discussion
in Section \ref{subsec:Non-observability-of-losing}). Note that one
can easily extend the analysis to the case where what is observed
ex post is a noisy signal about the ex post values, instead of the
ex post values themselves.} In the main part, we restrict attention to two-bidder auctions, but
we note that our main insights carry over when there are more bidders.

It should be highlighted that we allow for correlations between signals,
which will play a key role in the analysis. Moreover, despite the
correlation, the setting is one of private values, since the distribution
of bidder $i$'s ex post value is fully determined by bidder $i$'s
signal (i.e., it is unaffected by the other bidder's signal, conditional
on $i$'s signal). Yet, novice players are assumed to be unaware of
the true signal generating process, and thus of the private value
character of the auction. Instead, they construct a representation
of the statistical links between the variables of interest based on
the signal they receive as well as the dataset available to them.
Specifically, observations from past auctions take the form $(b_{1},v_{1},b_{2},v_{2})$
where $b_{j}$ is the bid previously submitted by a subject in the
role of bidder $j$ and $v_{j}$ is his ex post value. A novice bidder
$i$ constructs from the dataset the empirical distribution describing
how $b_{j}$ is distributed conditional on the various possible ex
post values of bidder $i$. He also uses his own signal $\theta_{i}$
about the likelihood of his various possible ex post values $v_{i}$
and combines the two to form a belief about how $(v_{i},b_{j})$ are
jointly distributed given her own signal. He then best-responds to
this belief given the rules of the auction.

We will be considering steady state environments in which there is
a mixed population of bidders composed of a share of novice bidders
(whose expectations are formed as just informally explained) and a
complementary share of rational bidders assumed to have a correct
understanding of the strategies of the two types of bidders as well
as their respective share. Observe that rational bidders can alternatively
be viewed as experienced bidders who would have had the opportunity
to find out the best strategy in their auction environment (without
necessarily an explicit knowledge of how the various types of bidders
behave nor of the shares of the various types). We will refer to such
steady states as Data-Driven Equilibria (see below for a discussion
of how this concept relates to other existing concepts).

We apply this model to understand the efficiency properties of Data-Driven
Equilibria, and more particularly, whether by a judicious choice of
auction rule, one can implement an efficient allocation.\footnote{Obviously, given that at the time of the auction, bidders do not know
their ex post values, efficiency is defined here at the interim stage,
based on the information available to bidders at the time of the auction.} This question resembles classic investigations in mechanism design
with the main difference that our solution concept is not the Bayes
Nash Equilibrium, but the Data-Driven Equilibrium designed to deal
with the presence of novice bidders in our environment. Another departure
from classic mechanism design is that for the main part of the paper,
we will not be considering abstract mechanisms with arbitrary messages
to be sent by bidders to the designer before an outcome (allocation
and transfer) is decided. Instead, we will focus on what we call auction-like
mechanisms defined as mechanisms in which each bidder submits a real-valued
bid, and an outcome is chosen as a function of the profile of bids
with the restriction that if a bidder submits a higher bid, this bidder
has more chance of winning the object.

Our insights are as follows. Unless the distributions of signals of
the two bidders are independent, data-driven bidders rely on a misspecified
statistical model and as a result choose suboptimal bidding strategies.
In Section \ref{sec:Standard-Auctions}, we start illustrating this
with Second-Price Auctions (SPA) in the (symmetric) binary case in
which there are two possible ex post values. We show that unlike rational
bidders, novice bidders do not bid their expected value when there
are correlations. As in winner's curse models, novice bidders make
inference about their ex post value from how the other bidder bids.
In the case of positive correlation, this leads novice bidders to
bid more than their expected value when they receive good signals
(because in the neighborhood of large opponent's bids, the own ex
post value is more likely to be high) and less than their expected
value when they receive bad signals (for a symmetric reason). We provide
a numerical characterization of the equilibrium for a parametric class
of distributions.

Clearly, the fact that novice and rational bidders do not bid in the
same way leads to inefficiencies, unless there is perfect correlation
of the signals, or the bidders are all novice or all rational. For
our parametric example, we observe that the normalized welfare loss
is U-shaped in the share of novice bidders as well as in the degree
of correlation. More generally, we show for the binary mixed population
case that as soon as there are correlations, there is some welfare
loss in the Second Price Auction. We also consider First-Price Auctions
(FPA), for which we also show that there must be inefficiencies whenever
there is correlation. We illustrate through an example that sometimes
the welfare loss may be larger in the Second-Price Auction than in
the First-Price Auction.

Our main result concerns general auction-like mechanisms when ex post
values can take at least three realizations. In Section \ref{sec:Auction-like-Mechanisms},
we provide a general inefficiency result. More precisely, we show
in the mixed population case that for generic joint distributions
of signals, there is no auction-like mechanism that allows to obtain
an efficient outcome with probability one in any Data-Driven Equilibrium.
The intuition for this result is as follows. To obtain efficiency
among rational bidders, only the Second-Price Auction or a strategically
equivalent auction format can be used because with more than two ex
post values there is generically a manifold signal realizations corresponding
to the same expected value for the object, but different beliefs about
the signal realization of the other bidder. Since in Second-Price
auctions, novice bidders do not bid their expected value as also observed
in the simplified binary case, we conclude that inefficiencies must
occur.

In Section \ref{sec:Discussion-and-Extensions}, we put our analysis
in perspective. First, we discuss alternative specifications of cognitive
limitations in auction-like mechanisms either due to different accessibility
to data sets from past auctions, or due to more or less sophisticated
use of the same limited data sets.\footnote{We note that the use of data as considered in the main model would
not allow to reconstruct the correct signal-generating process due
to a fundamental identification constraint. And we suggest our proposed
approach can be regarded as corresponding to a sophisticated use of
the limited dataset that does not make ad hoc assumptions about bidding
behavior in past actions.} We make the simple observation that our main impossibility result
would a fortiori hold if we were to consider a mixed population that
includes extra cognitive types in addition to those considered in
the main part of the paper.

Second, we discuss scenarios in which losing bids would not be accessible
from past auctions, thereby considering a natural further restriction
on the accessible data sets. We discuss various possible approaches
to modeling novice bidders in this context, and suggest that all of
them would lead to results similar to those obtained in our main model.

Third, we briefly discuss more general mechanisms beyond the auction-like
mechanisms and note that judicious use of such mechanisms (direct
mechanisms of the scoring rule type) may allow to elicit the beliefs
of every bidder $i$ about bidder $j$'s (interim) type. Since in
our baseline model, for generic distributions, no two different types
have the same belief about their opponent's type, such mechanisms
could potentially be used to implement a broad range of allocation
rules in the spirit of the work of \citet{Cremer1988}, \citet{Johnson1990},
\citet{McAfee1992}, and \citet{Gizatulina2017}. It should be mentioned
however that such mechanisms are not commonly used in market design
(perhaps because they require a level of knowledge -how the belief
about opponent's type relates to the valuation- that is rarely available
to the designer). This consideration has led us to restrict attention
to auction-like mechanisms which seem much more practical from a market
design perspective.\footnote{Moreover, it should be mentioned that if we were to consider richer
spaces of cognitive types as previously suggested, then one could
easily conceive that different types have the same belief, thereby
reducing the scope of what can be implemented with such abstract mechanisms
(even assuming the designer has the full knowledge required to make
use of such mechanisms).}

\subsection*{Related literature}

Our paper relates to different branches of literature. First, the
modeling of data-driven bidders is in the spirit of the Analogy-Based
Expectation Equilibrium \citep{Jehiel2005} to the extent that these
bidders aggregate the bid behavior of their opponent according to
their own ex post value. Such an aggregation of bidding behavior can
be related to the payoff-relevant analogy partition introduced in
\citet{Jehiel2008}. This modeling can also be related to the Bayesian
Network Equilibrium \citep{Spiegler2016}, viewing these agents as
believing that their values cause the bid of the opponent, but note
that here the reasoning of novice bidders is viewed as a consequence
of the nature of the dataset accessible to them, not as a consequence
of a subjective wrong causality relation they could have in mind (see
\citet{Spiegler2020}, and \citet{Jehiel2020}, for elaborations of
the link between the Analogy-Based Expectation Equilibrium and the
Bayesian Network Equilibrium).\footnote{At a more general level, one can also relate the Data-Driven Equilibrium
to the self-confirming equilibrium \citep{Battigalli1992,Dekel2004},
as well as recent behavior equilibrium models dealing with misspecified
beliefs \citep[see][among others]{Eyster2005,Esponda2008,Esponda2016}.}

Our paper is also related to the robust mechanism design literature
\citep{Bergemann2005}, in the sense that a common motivation in that
literature and our approach is that it may be hard to know what the
beliefs of agents are. While the robust mechanism design literature
uses this observation to motivate the desire to implement outcomes
for a large range of (or even all) beliefs, our paper explicitly suggests
a method of belief formation for bidders who do not have access to
such information from past auctions. Our paper is also mostly concerned
with a subclass of mechanisms that we refer to as auction-like mechanisms
and how these perform in the joint presence of data-driven bidders
and rational bidders, which has no counterpart in the literature on
robust mechanism design.

Finally, from a technical point of view, our analysis makes use of
some results developed in the literature on mechanism design with
correlation. In particular, we borrow genericity arguments from \citet{Gizatulina2017}.

\section{Model\label{sec:Model}}

\paragraph{Mechanisms}

We consider the allocation of a single object to two bidders $i=1,2$
via an auction or more general auction-like mechanism. To simplify
notation, when we consider a generic bidder $i\in\{1,2\}$, we denote
the opponent by $j\neq i$. A Mechanism $M=\left[(B_{i}),q,p\right]$
consists of three elements: (i) feasible bids $B_{i}$ for the two
bidders. A profile of bids is denoted $b=(b_{1},b_{2})\in B:=B_{1}\times B_{2}$.
(ii) an allocation rule $q:B\rightarrow[0,1]^{2}$, $q(b)=(q_{1}(b),q_{2}(b))$,
with $q_{1}(b)+q_{2}(b)\le1$, where $q_{i}(b)$ is the probability
that bidder $i$ gets the object if the bid profile $b$ is submitted.
(iii) A payment rule $t:B\rightarrow\mathbb{R}^{2}$, $p(b)=(p_{1}(b),p_{2}(b))$,
where $p_{i}(b)$ denotes the payment bidder $i$ has to make if the
bid profile $b$ is submitted.

\paragraph{Valuations}

Ex-post, the value of the object for bidder $i$ is denoted $v_{i}$.
It can take values in $V=\left\{ v^{1},\ldots v^{K}\right\} $. $v_{i}$
can also be interpreted as the expected value based on a signal that
is known ex-post if the ex-post value is not learned completely. Up
to normalization, it is without loss to assume that $0=v^{1}<\ldots<v^{K}=1$.
When participating in a mechanism, each bidder has an interim type
$\theta_{i}=(\theta_{i}^{1},\ldots,\theta_{i}^{K})\in\Theta:=\Delta V$,
where $\theta_{i}^{k}$ denotes the probability that $v_{i}=v^{k}$.
A profile of types is denoted $\theta=(\theta_{1},\theta_{2})$. We
assume that conditional on $\theta_{i}$, $v_{i}$ is independent
of $\theta_{j}$. As a consequence the expected valuation of a bidder
only depends on her own interim type: $E[v_{i}|\theta]=E[v_{i}|\theta_{i}]$.
In other words, we are considering a setting with \emph{private values}.
Interim types are jointly distributed with distribution function $F(\theta)$
and density $f(\theta)$ defined over\textbf{\ }$\Theta^{2}$,\textbf{\ }and
our main interest is in the case where $\theta_{1}$ and $\theta_{2}$
are not independent. We assume throughout that the joint distribution
is symmetric and has a continuous and positive density. When there
is no confusion, we slightly abuse notation and denote marginal distributions
$F_{i}(\theta_{i})$ and $f_{i}(\theta_{i})$ by $F(\theta_{i})$
and $f(\theta_{i})$; and conditional distributions $F_{i}(\theta_{i}|\theta_{j})$
and $f_{i}(\theta_{i}|\theta_{j})$, by $F(\theta_{i}|\theta_{j})$
and $f(\theta_{i}|\theta_{j})$.

\paragraph{Rational and Misspecified Bidders}

We assume that each bidder $i$ is characterized by a \emph{generalized
type} $t_{i}=\left(\theta_{i},s_{i}\right)$, where $\theta_{i}$
denotes the \emph{interim type} described before, and $s_{i}\in\left\{ r,m\right\} $
specifies the \emph{sophistication} of the bidder. We denote the set
of general types by $T=\Theta\times\{r,m\}$. For simplicity we will
call $\theta_{i}$ just the \emph{type}. The probability that $s_{i}=r$
is denoted $\lambda\in(0,1)$; we assume that it is independent of
$\theta_{i}$ and across bidders. $s_{i}=r$ means that bidder $i$
is \emph{rational}; and $s_{i}=m$ means that bidder $i$ is \emph{misspecified}.
Informally, the rational type correctly understands the environment,
whereas the misspecified type holds beliefs that are endogenously
determined by past observations of equilibrium outcomes of the mechanism
she currently participates in. As we will see, this way of forming
beliefs can lead to misspecifications, hence the name of the type.

We now make this precise. Fix a mechanism $M=\left[(B_{i}),q,p\right]$.
A strategy of bidder $i$ is a function $b_{i}:T\rightarrow B_{i}$,
where as a shorthand we write $b_{i}(\theta_{i},s_{i})=b_{i}^{s_{i}}(\theta_{i})$---that
is, $b_{i}^{r}(\cdot)$ is the strategy of the rational type, and
$b_{i}^{m}(\cdot)$ is the strategy of the misspecified type of bidder
$i$.\footnote{We only consider pure strategies in our setting with continuous interim
types.} A strategy profile is denoted by $b=(b_{1},b_{2})=(b_{1}^{r},b_{1}^{m},b_{2}^{r},b_{2}^{m})$
and we denote the space of all strategy profiles by $\mathcal{B}$.

For a rational type of bidder $i$, the expected utility of type $\theta_{i}$
when submitting bid $b_{i}\in B_{i}$, and assuming that bidder $j$
bids according to $b_{j}(\cdot)$, is given by 
\[
U_{i}^{r}(b_{i},\theta_{i}|b_{j}(\cdot))=\mathbb{E}_{f}\left[v_{i}\,q_{i}(b_{i},b_{j}(\theta_{j},s_{j}))-p_{i}(b_{i},b_{j}(\theta_{j},s_{j}))|\theta_{i}\right],
\]
where $\mathbb{E}_{f}$ is the expectation with respect to the correct
distribution $f$ and the probability $\lambda$.

Next consider the misspecified type. We assume that this type forms
a belief using past observations from the same mechanism played by
similar bidders. Suppose the mechanism is run repeatedly with two
(short-lived) bidders whose generalized type profiles are drawn i.i.d.,~across
repetitions. If both bidders play according to a fixed strategy profile,
as they would in a steady state, then repeated play generates a data
set with observations $(b_{1},v_{1},b_{2},v_{2})$. We make the assumption
that only bids and ex-post valuations are observable.
\begin{assumption}
\label{assu:Observability}For each mechanism we consider, we assume
that bidders have access to observations of the form $(b_{1},v_{1},b_{2},v_{2})$
from the same mechanism. The data about past mechanisms does not include
the types $(\theta_{1},\theta_{2})$ of past bidders.
\end{assumption}
\noindent The idea behind this assumption is that bids are often disclosed
after an auction and as time goes by, the ex-post valuation of the
bidders, or an estimate thereof becomes known as well. On the other
hand, bidders typically do not have access to the beliefs that past
bidders in their role held at the time of bidding.

Past data allow bidders to identify the joint distribution of observable
variables. We abstract from issues of estimation, and assume that
bidders can recover this distribution without estimation error. The
misspecified bidder then forms a simple model that combines relevant
information from the empirical distribution of $(b_{1},v_{1},b_{2},v_{2})$,
and her belief that her own $v_{i}$ is distributed according to $\theta_{i}$.
To illustrate consider an auction with possible bids $B_{1}=B_{2}=[0,\infty)$.
To assess the payoff from different bids, a bidder needs to know the
joint distribution of her own valuation $v_{i}$ and the opponent's
bid $b_{j}$, conditional on her own type $\theta_{i}$. The misspecified
bidder combines the distribution of $v_{i}$ given by her type $\theta_{i}$
with the joint distribution of $v_{i}$ and the opponent's bid $b_{j}$
learned from the data in a parsimonious way, taking the joint distribution
to be 
\begin{equation}
\mathbb{P}_{m}\left[v_{i}=v^{k},b_{j}\leq b\middle|\theta_{i}\right]=\theta_{i}^{k}\times H_{i}(b|v^{k})\label{eq:misspecified_probabilitiy}
\end{equation}
where $H_{i}(b|v^{k})$ is the c.d.f.~of $b_{j}$ conditional on
$v_{i}=v^{k}$ that is obtained from the data. Throughout, we will
use $\mathbb{P}_{m}$ for probabilities assessed by the misspecified
type and $\mathbb{P}_{f}$ for probabilities computed using the correct
probabilistic model (given the density ``$f$''). To see the difference
in this particular example, note that 
\[
\mathbb{P}_{f}\left[v_{i}=v^{k},b_{j}\leq b\middle|\theta_{i}\right]=\theta_{i}^{k}\times\mathbb{P}_{f}\left[b_{j}\leq b\middle|\theta_{i},v_{i}=v^{k}\right]=\theta_{i}^{k}\times\mathbb{P}_{f}[b_{j}\leq b|\theta_{i}]
\]
where the second equality follows from the assumption that $\theta_{j}$
and $v_{i}$ are independent, conditional on $\theta_{i}$. Under
Assumption \ref{assu:Observability}, $\mathbb{P}_{f}[b_{j}\leq b|\theta_{i}]$
cannot be assessed directly from the data since the types of past
bidders are not available. In order to identify $\mathbb{P}_{f}[b_{j}\leq b|\theta_{i}]$
from the data, one would have to make assumptions about the strategies
used by past bidders. These assumptions are ad hoc if only data on
past bids and ex-post signals are available and a misspecified bidder
does not have insight into the type of reasoning used by past bidders.
The misspecified type therefore does not attempt to use the data through
the lens of such assumptions but just takes the empirical correlation
between $v_{i}$ and $b_{j}$ as given. One way to interpret the difference
between the rational and the misspecified type is that the former
is an experienced bidder who understands how bidders reason and is
thus able to formulate the correct model to interpret the data where
the latter lacks this understanding.\footnote{An alternative equivalent interpretation of a rational bidder is that
such a bidder has found out the optimal strategy (possibly through
a trial and error process) whereas the misspecified bidder can only
rely on the publicly available data from past similar auctions.} In particular, the misspecified type does not know that conditional
on her own type $\theta_{i}$, her own valuation is independent of
the opponent's type and thus her valuation and the opponents bid are
also conditionally independent. The data available from past auctions,
however, exhibits a correlation between the valuation $v_{i}$ and
the bid $b_{j}$, since conditioning on the the unobserved type $\theta_{i}$
is not possible. This is the source of the misspecification of the
$m$-type. As we will see, this gives rise to bidding behavior that
is similar to the winner's curse.

To summarize, for a misspecified type of bidder $i$, the expected
utility of type $\theta_{i}$ when submitting bid $b_{i}\in B_{i}$,
and assuming that bidder $j\neq i$ bids according to $b_{j}(\cdot)$,
is given by 
\begin{align*}
U_{i}^{m}(b_{i},\theta_{i}|b_{j}(\cdot)) & =\mathbb{E}_{m}\left[v_{i}\,q_{i}(b_{i},b_{j}(\theta_{j},s_{j}))-p_{i}(b_{i},b_{j}(\theta_{j},s_{j}))|\theta_{i}\right],\\
 & =\sum_{k=1}^{K}\theta_{i}^{k}\int_{B_{j}}[v^{k}\,q_{i}(b_{i},b_{j})-p_{i}(b_{i},b_{j})]dH_{i}(b_{j}|v^{k}),
\end{align*}
where $\mathbb{E}_{m}$ is the expectation formed according to the
model described above. Note that in order to determine $H_{i}(\cdot|v^{k})$,
it is enough to specify the strategy $b_{j}(\cdot)$ since $v_{i}$
and $b_{j}$ in the current auction do not depend on the bids placed
by the bidder in role $i$ in the past.\footnote{Another way to motivate why bidder $i$\ does not use the $b_{i}$\ from
past auctions is that she is unsure what led a past bidder $i$\ to
choose his bid, as it could be determined by his information and/or
his way of reasoning none of which are accessible to her.} To understand this better, the example of the second-price auction
in the next section will be helpful.

\paragraph{Equilibrium}

To close the model, we assume that $H_{i}(\cdot|v^{k})$ are equilibrium
objects that are generated by the equilibrium strategy profile and
the misspecified type best-responds given her beliefs that are captured
by $H_{i}(\cdot|v^{k})$.
\begin{defn}
The strategy profile $b(\cdot)$ is a \emph{``Data-Driven}\textbf{\emph{\ }}\emph{Equilibrium''}
of the mechanism $M=\left[(B_{i}),q,t\right]$ if for all $i\neq j$,
and for all $\theta_{i}\in\Theta$, 
\end{defn}
\begin{enumerate}
\item $b_{i}^{r}(\theta_{i})\in\arg\max_{b_{i}\in B_{i}}U_{i}^{r}(b_{i},\theta_{i}|b_{j}(\cdot))$,
\item $b_{i}^{m}(\theta_{i})\in\arg\max_{b_{i}\in B_{i}}U_{i}^{m}(b_{i},\theta_{i}|b_{j}(\cdot))$,
where the distribution $H_{i}(b_{j}|v^{k})$ used to compute $U_{i}^{m}$
is derived from $\beta$, $f$, and $\text{Prob}[s_{i}=r]=\lambda$. 
\end{enumerate}

\section{Standard Auctions\label{sec:Standard-Auctions}}

Before considering auction-like mechanisms and presenting the main
result of the paper, we apply the model to standard auctions. This
illustrates how data-driven beliefs affect bidding behavior.

In order to apply the model to standard auctions, we consider the
case where $|V|=2$, so that the type of each bidder is one-dimensional.
More specifically, we assume that $V=\{0,1\}$, so that the type can
be written as one number $\theta_{i}\in[0,1]$, that specifies the
probability that bidder $i$'s ex-post valuation is $v_{i}=1$. Note
that this implies that $\theta_{i}$ is also the interim expected
value of bidder $i$. In the following, we explain the equilibrium
logic of our model for two standard auctions formats, the Second-Price
Auction and the First-Price Auction. To compute concrete bidding equilibria,
we will use a parametric class of joint distributions that allows
us to vary the correlation between $\theta_{1}$ and $\theta_{2}$.
\begin{example}
\label{exa:joint-density}The joint density is given by 
\[
f(\theta_{1},\theta_{2})=\frac{2+\alpha}{2}\left(1-\left\vert \theta_{1}-\theta_{2}\right\vert \right)^{\alpha}.
\]
The parameter $\alpha\in[0,\infty)$ determines the correlation between
the two types where $\alpha=0$ corresponds to the independent case
and $\alpha=\infty$ corresponds to perfect correlation. Figure \ref{fig:Joint-density}
depicts the joint density for $\alpha\in\{.1,1,10\}$. 
\begin{figure}[h]
\includegraphics[width=1\textwidth]{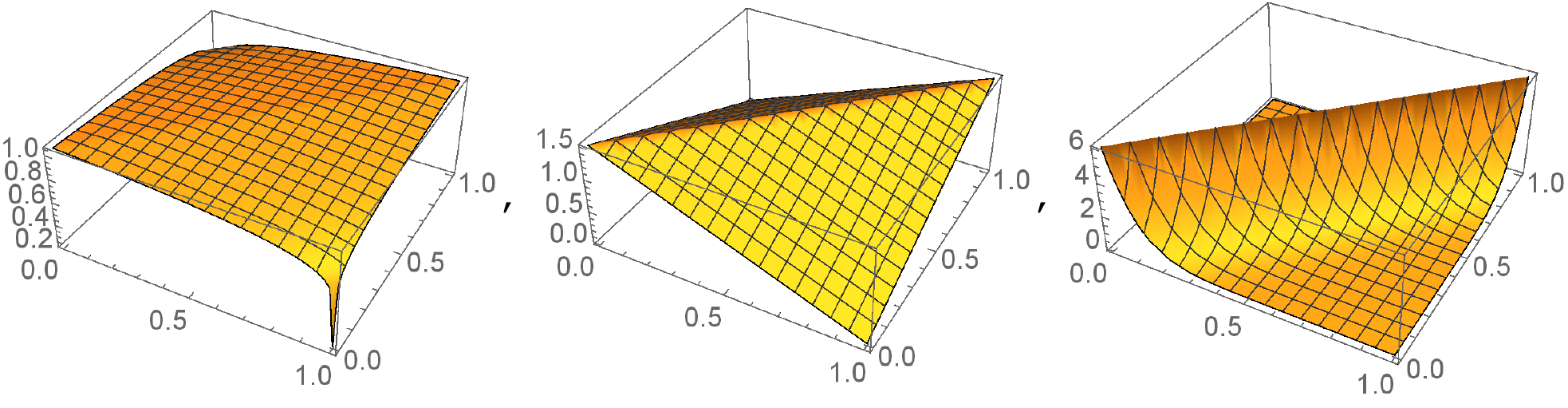} \caption{Joint density, $\alpha\in\{.1,1,10\}$ (left to right) }
\label{fig:Joint-density} 
\end{figure}
\end{example}

\subsection{Second-price Auction}

In a second-price auction, the rational type has a weakly dominant
strategy since values are private. Hence she bids her interim expected
value. We have 
\[
b^{r}(\theta_{i})=\theta_{i},
\]
where $b^{r}$ refers to the rational type's strategy. We denote the
inverse by $\theta^{r}(b_{i})$, which is of course equal to\textbf{\ }$b_{i}$
in this case.

Now consider the misspecified type and consider a symmetric equilibrium,
that is $b_{i}^{m}(\cdot)=b_{j}^{m}(\cdot)=b^{m}(\cdot)$. Suppose
the equilibrium strategy $b^{m}(\cdot)$ is strictly increasing with
inverse $\theta^{m}(b_{i})$. In equilibrium, the distribution of
$b_{j}$ conditional on $v_{i}=1$ is 
\begin{equation}
H^{\text{SPA}}(b\mid v_{i}=1)=\frac{\mathbb{P}_{f}\left[b_{j}\leq b,v_{i}=1\right]}{\mathbb{P}_{f}\left[v_{i}=1\right]},\label{eq:HSPA1_informal}
\end{equation}
where $H^{\text{SPA}}(\cdot)$ refers to this distribution for the
SPA. Note that the misspecified type learns the correct joint distribution
of $v_{i}$ and $b_{j}$ from the data. Hence we have used the correct
probabilities $\mathbb{P}_{f}$ on the right-hand side. In the denominator,
we have the unconditional probability of $v_{i}=1$ which is given
by the (ex-ante) expectation of the random variable $\tilde{\theta}_{i}$.
In the numerator, the probability $\mathbb{P}_{f}[b_{j}\leq b,v_{i}=1]$
is obtained by averaging $\mathbb{P}_{f}[b_{j}\leq b,v_{i}=1|\tilde{\theta}_{i}]$
over the (ex-ante) random variable $\tilde{\theta}_{i}$. Since $b_{j}$
is a function of $\theta_{j}$ and $s_{j}$, and the \emph{generalized
type} $(\theta_{j},s_{j})$ and $v_{i}$ are independent \emph{conditional
on $\tilde{\theta}_{i}$}, we have: 
\begin{align*}
H^{\text{SPA}}(b\mid v_{i}=1) & =\frac{\mathbb{E}_{\tilde{\theta}_{i}}\left[\mathbb{P}_{f}[b_{j}\leq b|\tilde{\theta}_{i}]\times\mathbb{P}_{f}[v_{i}=1|\tilde{\theta}_{i}]\right]}{\mathbb{E}[\tilde{\theta}_{i}]}\\
 & =\frac{\mathbb{E}_{\tilde{\theta}_{i}}\left[\left(\lambda\mathbb{P}_{f}[b^{r}(\theta_{j})\leq b|\tilde{\theta}_{i}]+(1-\lambda)\mathbb{P}_{f}[b^{m}(\theta_{j})\leq b|\tilde{\theta}_{i}]\right)\times\mathbb{P}_{f}[v_{i}=1|\tilde{\theta}_{i}]\right]}{\mathbb{E}[\tilde{\theta}_{i}]}\\
 & =\frac{1}{\mathbb{E}[\tilde{\theta}_{i}]}\int_{0}^{1}\left[\lambda F(b\mid\tilde{\theta}_{i})+(1-\lambda)F(\theta^{m}(b)\mid\tilde{\theta}_{i})\right]\tilde{\theta}_{i}\;f(\tilde{\theta}_{i})d\tilde{\theta}_{i}.
\end{align*}
In the second line we decomposed the probability $\mathbb{P}_{f}[b_{j}\leq b|\tilde{\theta}_{i}]$
into the probability that a rational and a misspecified type bid below
$b$, conditional on $\tilde{\theta}_{i}$. If the opponent is rational,
the probability of $b_{j}\leq b$ is given by $\mathbb{P}_{f}[b^{r}(\theta_{j})\leq b|\tilde{\theta}_{i}]=F(\theta^{r}(b)\mid\tilde{\theta}_{i})=F(b\mid\tilde{\theta}_{i})$,
and if the opponent is misspecified it is given by $\mathbb{P}_{f}[b^{m}(\theta_{j})\leq b|\tilde{\theta}_{i}]=F(\theta^{m}(b)\mid\tilde{\theta}_{i})$.
The term $\tilde{\theta}_{i}$\ in the third line is just $\mathbb{P}_{f}[v_{i}=1|\tilde{\theta}_{i}]$.
We obtain a similar expression for the distribution of $b_{j}$ conditional
on $v_{i}=0$: 
\[
H^{\text{SPA}}(b\mid v_{i}=0)=\frac{1}{\mathbb{E}[1-\tilde{\theta}_{i}]}\int_{0}^{1}\left[\lambda F(b\mid\tilde{\theta}_{i})+(1-\lambda)F(\theta^{m}(b)\mid\tilde{\theta}_{i})\right](1-\tilde{\theta}_{i})\;f(\tilde{\theta}_{i})d\tilde{\theta}_{i}
\]
where the expectation in the integral differs from that in $H^{\text{SPA}}(b\mid v_{i}=1)$
since $P_{f}[v_{i}=0|\tilde{\theta}_{i}]=(1-\tilde{\theta}_{i})$,
and outside the integral $\mathbb{E}[1-\tilde{\theta}_{i}]$ is the
unconditional probability $P_{f}[v_{i}=0]$.\footnote{Note that $H^{\text{SPA}}(b\mid v_{i}=0)=\mathbb{P}_{f}\left[b_{j}\leq b,v_{i}=0\right]/\mathbb{P}_{f}\left[v_{i}=0\right]$}

In a symmetric equilibrium of the second-price auction, the misspecified
type's bid for $\theta_{i}$ solves 
\[
\max_{b}\left\{ \theta_{i}H^{\text{SPA}}(b\mid v_{i}=1)-\theta_{i}\int_{0}^{b}xdH_{{}}^{\text{SPA}}(x|v_{i}=1)-(1-\theta_{i})\int_{0}^{b}xdH^{\text{SPA}}(x|v_{i}=0)\right\} 
\]
To obtain an equilibrium we have to determine a bidding strategy $b^{m}$
and the implied $H^{\text{SPA}}$ such that $b^{m}$ is optimal for
the misspecified type give belief $H^{\text{SPA}}$. Taking the first-order
condition for $b$ and substituting $H^{\text{SPA}}(b\mid v_{i}=1)$
and $H^{\text{SPA}}(b\mid v_{i}=0)$, we obtain a differential equation
for $b^{m}$.

In Example \ref{exa:joint-density}, when $\alpha=0$---the independent
case---we have that $H^{\text{SPA}}(b\mid v_{i}=1)=H^{\text{SPA}}(b\mid v_{i}=0)$
and the first order condition leads to $b^{m}(\theta_{i})=\theta_{i}$.
But, when $\alpha$\ differs from $0$, $H^{\text{SPA}}(b\mid v_{i}=1)$
differs from $H^{\text{SPA}}(b\mid v_{i}=0)$ and $b^{m}(\theta_{i})$
differs from $\theta_{i}$. Solving the differential equation numerically
for the joint distribution from Example \ref{exa:joint-density},
we get the bid-functions illustrated in Figure \ref{fig:Bid_functions_spa}.
\begin{figure}
\noindent \begin{centering}
\includegraphics[clip,height=0.14\textheight,viewport=0bp 0bp 590.198bp 371bp]{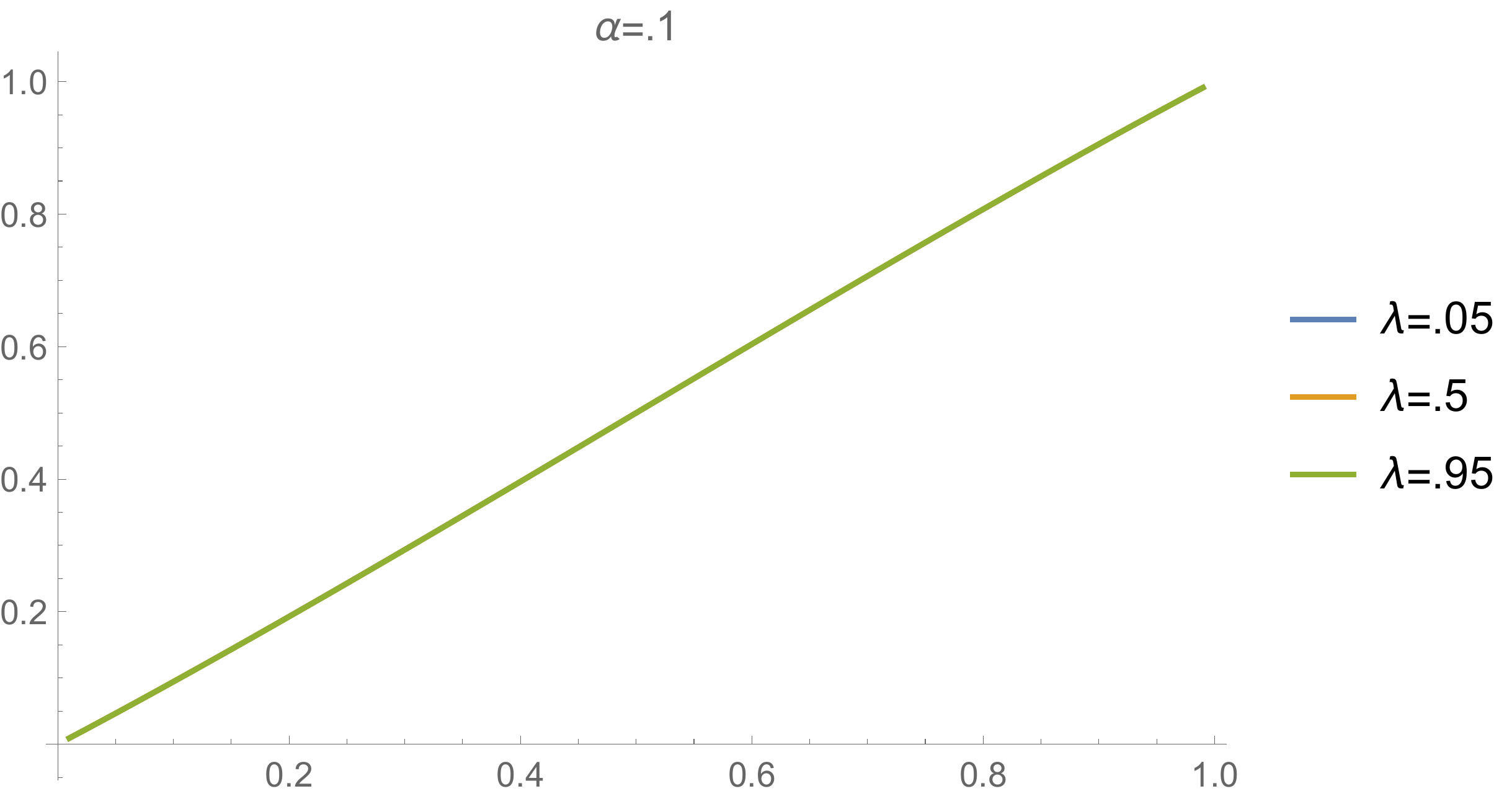}\includegraphics[clip,height=0.14\textheight,viewport=0bp 0bp 590.198bp 371bp]{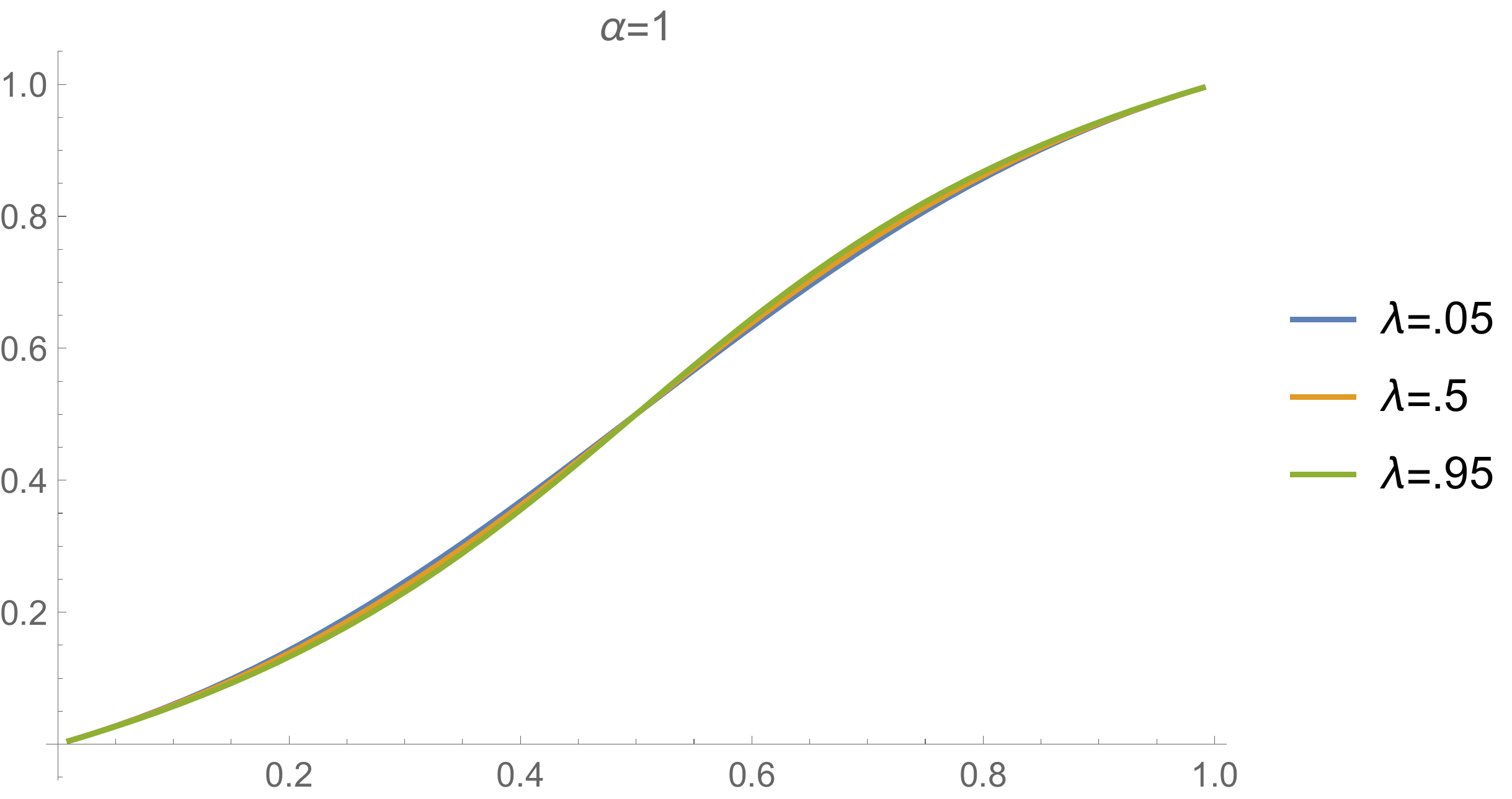}\includegraphics[clip,height=0.14\textheight]{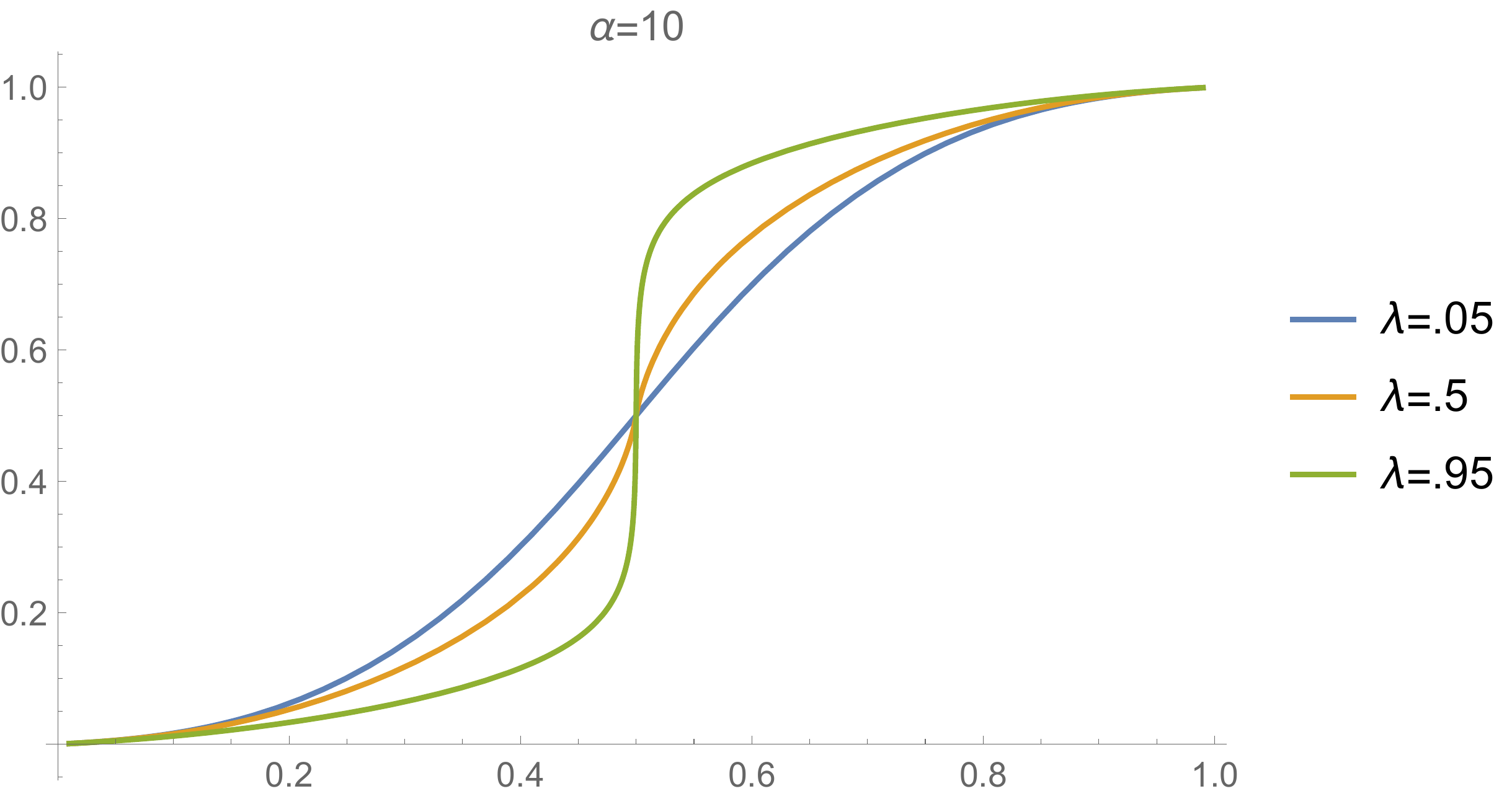} 
\par\end{centering}
\noindent \caption{SPA bid-function $b_{m}(\theta_{i})$, $\alpha\in\{.1,1,10\}$ (left
to right) }
\label{fig:Bid_functions_spa} 
\end{figure}

We see that increasing the correlation leads to stronger deviations
from the rational bid. Moreover, the sensitivity of $b_{m}$ with
respect to $\lambda$ becomes stronger if the correlation is stronger.
Generally, for fixed correlation, increasing the share of misspecified
types $(1-\lambda)$ leads to smaller deviation from rationality.
Bidding against mainly rational types, a misspecified type's behavior
exhibits strong deviations from rationality,\footnote{Numerical computations indicate that even if $\lambda\rightarrow1$,
the slope of $b^{m}$ remains bounded, where the bound depends on
$\alpha$. In other words, $b^{m}$ does \emph{not} converge to a
step function according to the numerical results.} but in equilibrium, the presence of other misspecified types has
a dampening effect.

\paragraph{Intuition}

The reasoning leading to the derivation of $b^{m}$ follows a logic
similar to that in classic analysis of winner's curse models. We observe
from Figure \ref{fig:Bid_functions_spa} that the misspecified type
overbids for $\theta_{i}>1/2$ and underbids for $\theta_{i}<1/2$.
What explains this behavior? To understand this, it is useful to shut
down the (dampening) equilibrium effect of misspecified types and
assume that $\lambda\approx1$. The crucial observation is that the
$m$-type believes that conditional on $v_{i}=1$, the opponent's
bid distribution is strong. This is because in the data, $v_{i}$
and $b_{j}$ are positively correlated: Observations with $v_{i}=1$
are more likely generated when $\tilde{\theta}_{i}$ is high. Due
to the positive correlation between $\theta_{i}$ and $\theta_{j}$,
this implies that $b_{j}$ is also likely to be high. Conversely,
the $m$-type believes that conditional on $v_{i}=0$, the opponent's
bid distribution is weak.

For an $m$-type with low $\theta_{i}$, consider the incentives to
decrease the bid below $b=\theta_{i}$. In this range reducing the
bid has a large effect on the winning probability conditional on $v_{i}=0$
(the $m$-type believes that conditional on $v_{i}=0$, the opponents
bid's are concentrated on a low range) and little effect on the winning
probability conditional on $v_{i}=1$ (where the $m$-type believes
the opponents bid's are concentrated on a high range). Therefore,
the $m$-type believes that by shading the bid, she can cut the losses
from winning with $v_{i}=0$, without a strong reduction of the gains
from winning when $v_{i}=1$.

For a high $\theta_{i}$, this logic is reversed. Consider the incentives
to increase the bid above $b=\theta_{i}$ when $\theta_{i}$ is high.
The bid is now in a range where the $m$-type believes that increasing
the bid mainly effects the winning probability conditional on $v_{i}=1$
and has little effect on the winning probability conditional on $v_{i}=0$.
Hence, she thinks overbidding increases the profits from winning with
$v_{i}=1$, while only modestly increasing the losses from winning
with $v_{i}=0$. This leads to bids above $\theta_{i}$ for high types
of the misspecified bidder.

\paragraph{Inefficiency of the Second-Price Auction}

While the distortions observed in the example are specific to the
parametric class of distributions, we can show generally that the
SPA is not efficient whenever both rational and misspecified types
arise with positive probability, and the types of the two bidders
are correlated.\footnote{Correlation is a sufficient condition for an inefficiency. The careful
reader will see from the proof that weaker forms of dependency also
lead to inefficiencies. In Section \ref{sec:Auction-like-Mechanisms}
we generalize this proposition to any finite number of valuations
(see Lemma \ref{lem:m-type_dont_bid_truthfully}).}
\begin{prop}
\label{prop:SPA_is_inefficient}If $\lambda\in(0,1)$ and $\text{Corr}[\theta_{1},\theta_{2}]\neq0$,
then any equilibrium of the second-price auction in which the rational
types of both bidder play their dominant strategies is inefficient. 
\end{prop}
\begin{proof}
All omitted proofs can be found in Appendix \ref{sec:Omitted-Proofs}. 
\end{proof}

\paragraph{Revenue and Efficiency}

Continuing our illustration for the parametric class in Example \ref{exa:joint-density},
we show how revenue and (relative) efficiency of the allocation varies
with (a) the share of rational types $\lambda$ and (b) the correlation
between $\theta_{1}$ and $\theta_{2}$---that is the parameter $\alpha$.

Figure \ref{fig:spa_revenue} plots the revenue as a function of $\lambda$
for different values of $\alpha$. Note that the comparison between
different values of $\alpha$ with $\lambda$ held fixed is not very
informative since the joint distribution changes in a complicated
way as $\alpha$ changes. 
\begin{figure}
\noindent \begin{centering}
\includegraphics[width=0.45\textwidth]{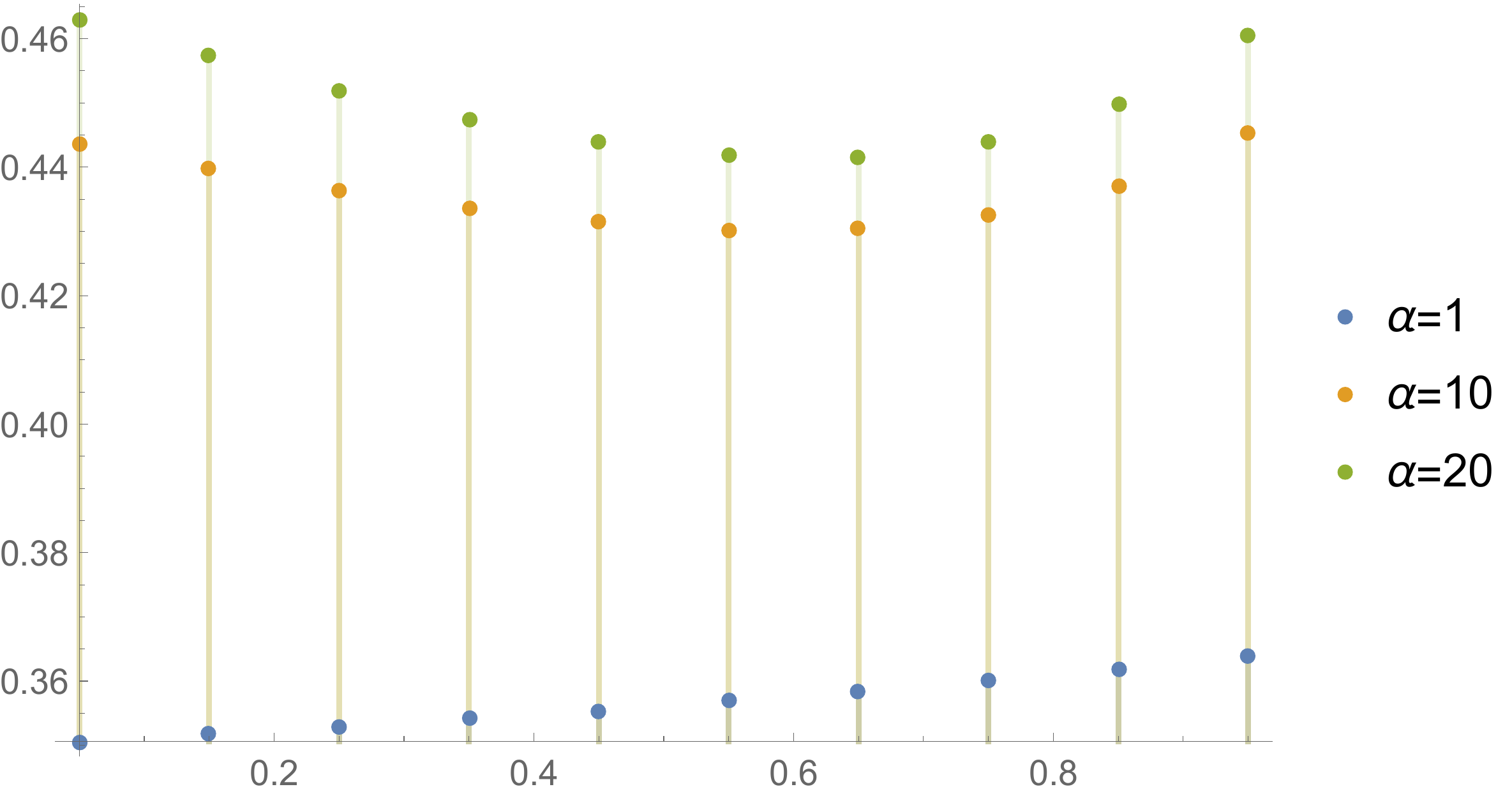} 
\par\end{centering}
\caption{Revenue from the SPA as a function of $\lambda$: for $\alpha\in\{1,10,20\}$}
\label{fig:spa_revenue} 
\end{figure}

We see that for the case of weak correlation ($\alpha=1$), revenue
is increasing in the share of rational type. This suggests that the
distortions in the misspecified type's bidding function adversely
affect revenue. For highly correlated interim types, the pattern changes
and revenue is U-shaped in the share of rational types. The initial
decline is intuitive since the distortions in the $m$-types bid become
larger if the share of rational types increases. Profits rise again
if the share of rational types becomes so large that presence of $m$-types
becomes unlikely.

Figure \ref{fig:spa_eff} shows how efficiency changes depending on
$\lambda$ and $\alpha$. 
\begin{figure}
\noindent \begin{centering}
\includegraphics[width=0.45\textwidth]{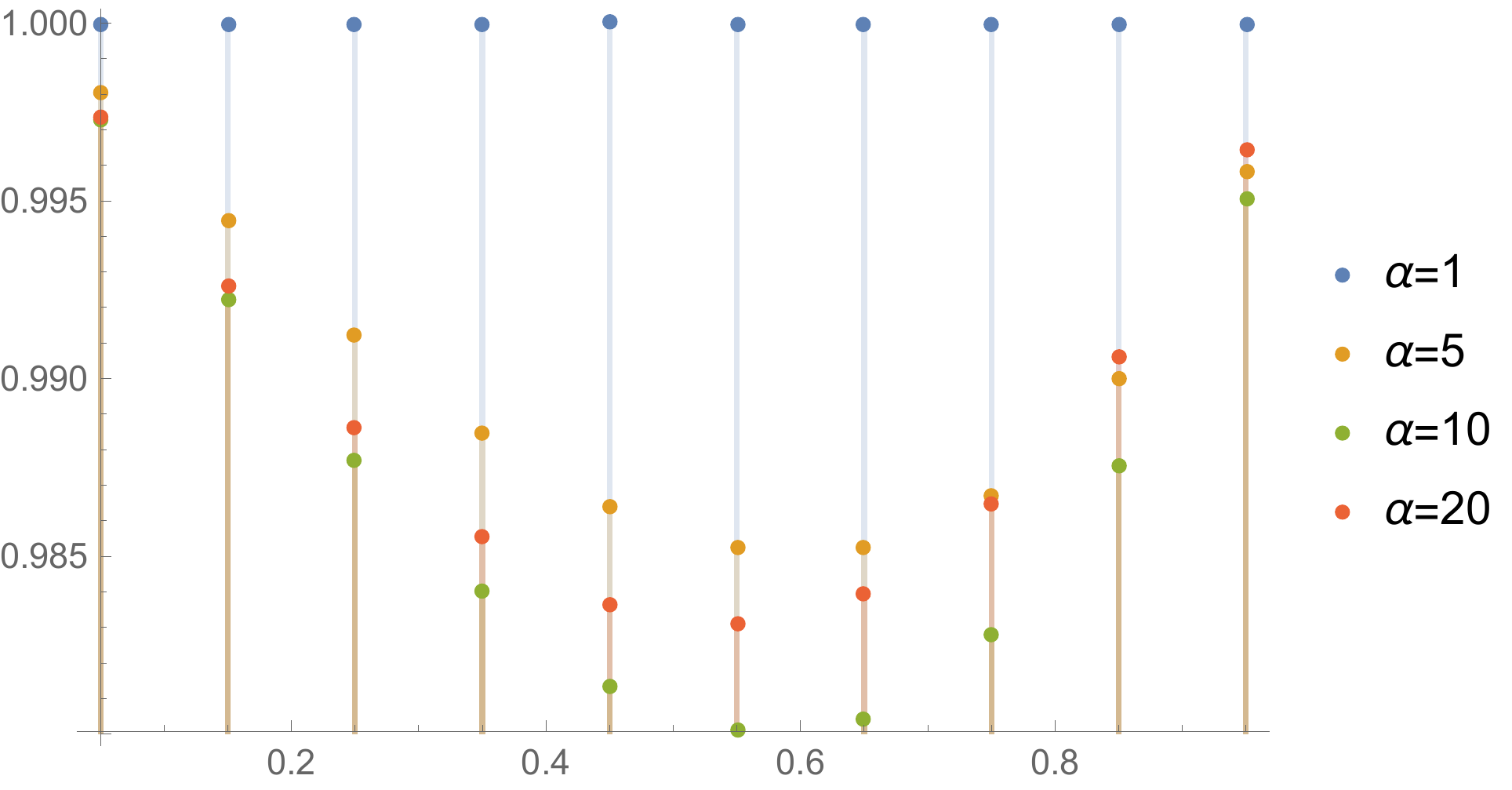} 
\par\end{centering}
\caption{Efficiency of SPA as a function of $\lambda$: for $\alpha\in\{1,5,10,20\}$}
\label{fig:spa_eff} 
\end{figure}

To make this comparable across different parameter sets, we normalize
efficiency by the expected ex-post value achieved if the object is
always allocated to the bidder with the highest interim type. Clearly
when $\lambda=0$\ or $1$, there is no inefficiency given that bidders
of the same sophistication bid in the same way. Moreover, both when
$\alpha=0$\ (the independent case) or $\alpha=\infty$ (perfect
correlation) there is no inefficiency either. In the parametric example,
we observe that the relative efficiency is $U$-shaped as a function
of $\lambda$\ and $\alpha$, as shown in Figure \ref{fig:spa_eff-1}.
\begin{figure}
\noindent \begin{centering}
\includegraphics[width=0.45\textwidth]{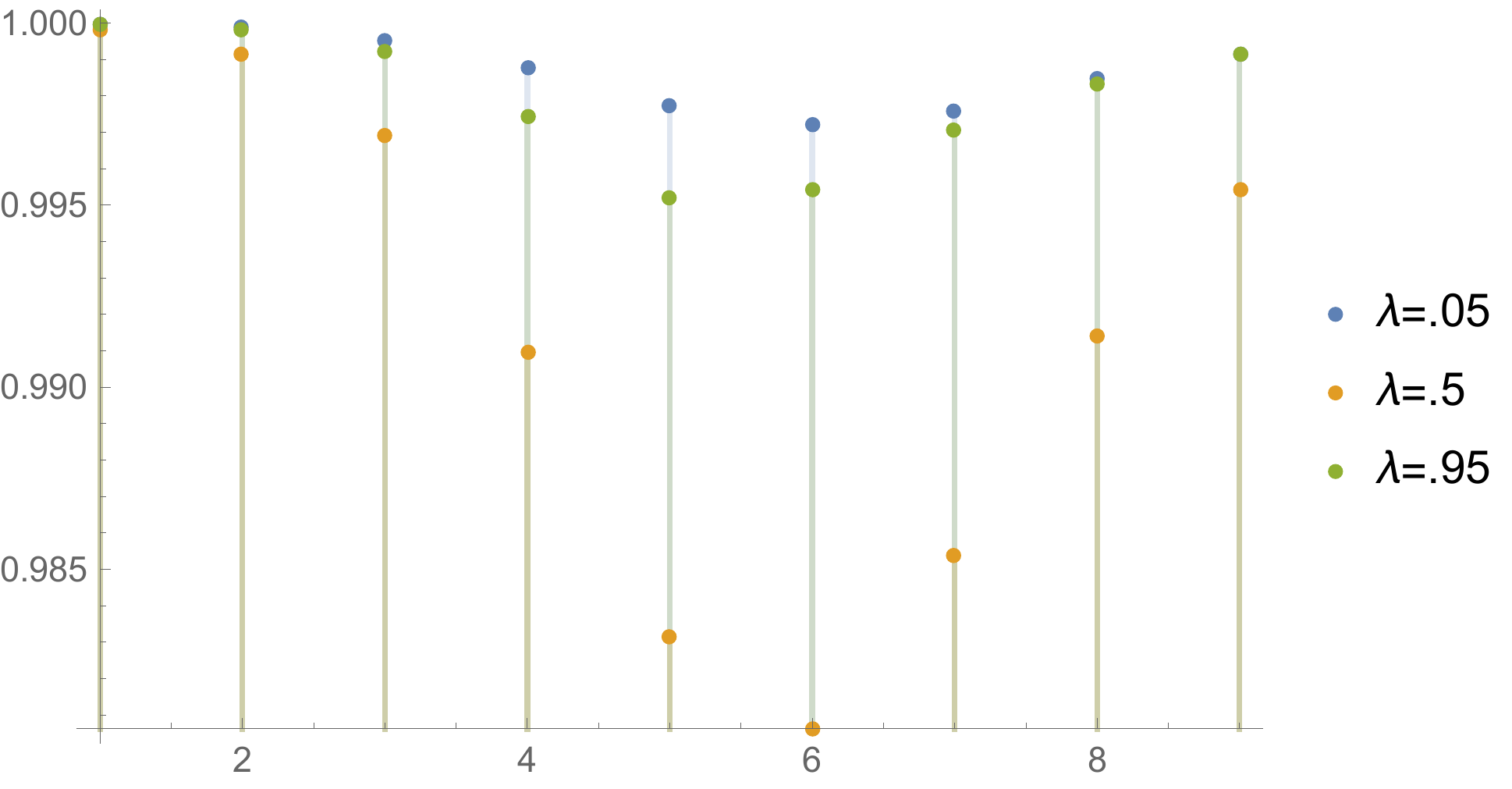} 
\par\end{centering}
\caption{Efficiency of SPA as a function of $\alpha=1/5+5^{k}$ where $k=1,\ldots,9$
is on the horizontal axis. $\lambda\in\{.05,.5,.95\}$. }
\label{fig:spa_eff-1} 
\end{figure}

\subsection{First-price auction}

In a first price auction, we obtain the misspecified type's belief
in similar way as for the second price auction: 
\begin{align*}
H^{\text{FPA}}(b\mid v_{i}=1) & =\int_{0}^{1}\left[\lambda F(\theta^{r}(b)\mid\tilde{\theta}_{i})+(1-\lambda)F(\theta^{m}(b)\mid\tilde{\theta}_{i})\right]\tilde{\theta}_{i}\frac{f(\tilde{\theta}_{i})}{E[\tilde{\theta}_{i}]}d\tilde{\theta}_{i},\\
H^{\text{FPA}}(b\mid v_{i}=0) & =\int_{0}^{1}\left[\lambda F(\theta^{r}(b)\mid\tilde{\theta}_{i})+(1-\lambda)F(\theta^{m}(b)\mid\tilde{\theta}_{i})\right](1-\tilde{\theta}_{i})\frac{f(\tilde{\theta}_{i})}{E[1-\tilde{\theta}_{i}]}d\tilde{\theta}_{i}.
\end{align*}
$b^{r}(\cdot)$ and $b^{m}(\cdot)$ now denote the bidding strategies
of the rational and misspecified types in the symmetric equilibrium
of the FPA, and their inverses are denoted by $\theta^{r}(\cdot)$
and $\theta^{m}(\cdot)$. The misspecified bidder's bid for type $\theta_{i}$
maximizes 
\begin{equation}
\max_{b}\left(1-b\right)\theta_{i}H^{\text{FPA}}(b\mid v_{i}=1)-b(1-\theta_{i})H^{\text{FPA}}(b|v_{i}=0).\label{eq:m-problem_FPA}
\end{equation}
Again we obtain a differential equation for $b^{m}(\theta_{i})$.
In contrast to the second price auction, however, we cannot assume
that rational bidders bid their expected valuations. Instead they
maximize 
\[
\max_{b}\left(\theta_{i}-b\right)\left(\lambda F(\theta^{r}(b)|\theta_{i})+(1-\lambda)F(\theta^{m}(b)|\theta_{i})\right)
\]
This optimization problem reflects the complete awareness of the model
by rational bidders. They use the correct distribution $f$, the share
of rational types in the population, and the equilibrium bidding strategies
of both the rational and the misspecified types when determining their
optimal bids. The first-order condition for the rational type's problem
yields a second differential equation. To compute an equilibrium,
we need to solve the system of two ODEs with the boundary condition
$(b^{m}(0),b^{r}(0))=(0,0)$. This proves challenging even for the
distributions in our example, since the system has a singular point
at the boundary condition. However, we obtain a similar inefficiency
result as we had for the SPA.
\begin{prop}
\label{prop:FPA_is_inefficient}If $\lambda\in(0,1)$ and $\text{Corr}[\theta_{1},\theta_{2}]\neq0$,
then the symmetric equilibrium of the first-price auction is inefficient. 
\end{prop}

\subsection{Comparison}

We can compute the bidding equilibrium for both auction formats for
the case of only rational bidders ($\lambda=1$) and only misspecified
bidders $\lambda=0$. Figure \ref{fig:Comparison_SPA_FPA} shows the
bid functions $b_{k}^{s}$ where $k=1,2$ denotes first- or second-price
auctions and $s=m,r$ denotes the misspecified or rational type. 
\begin{figure}
\begin{centering}
\includegraphics[height=0.2\textheight]{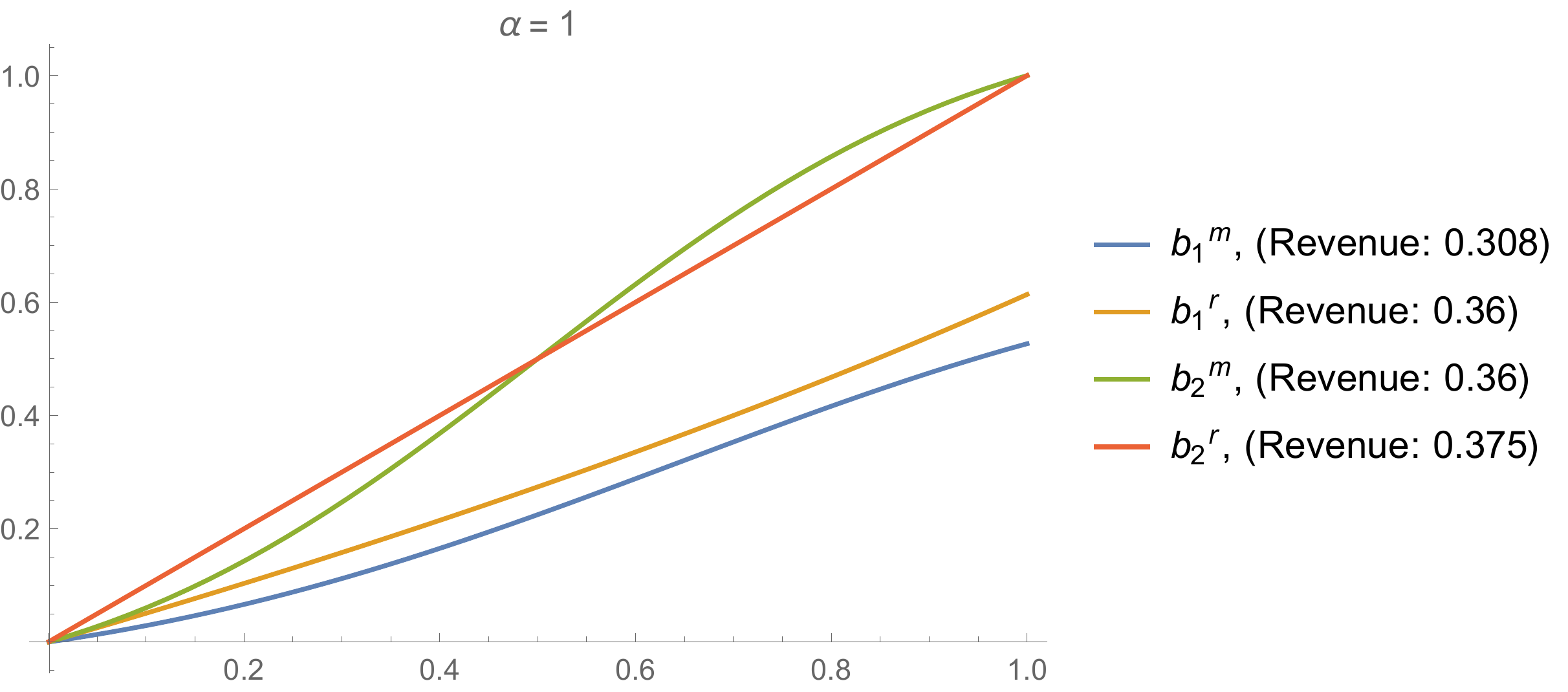} 
\par\end{centering}
\begin{centering}
\includegraphics[height=0.2\textheight]{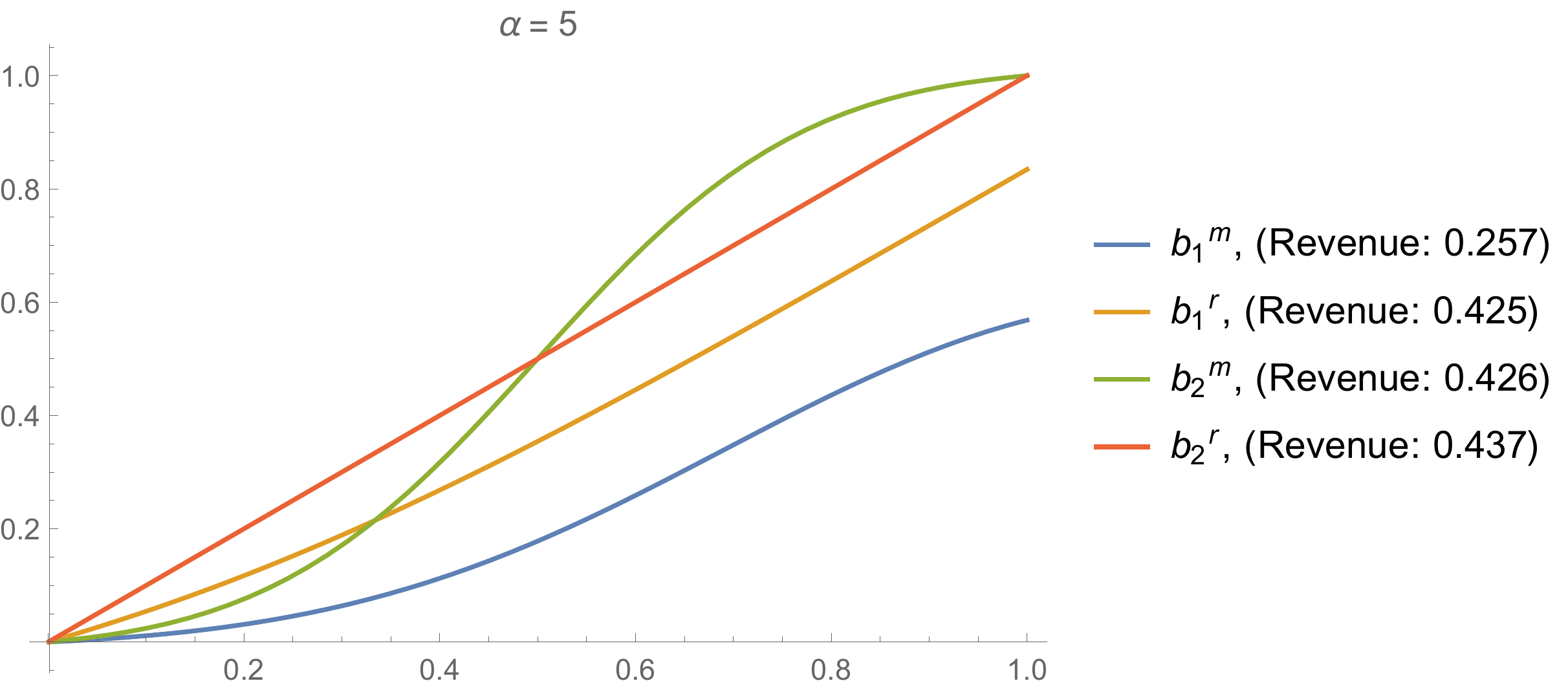} 
\par\end{centering}
\begin{centering}
\includegraphics[height=0.2\textheight]{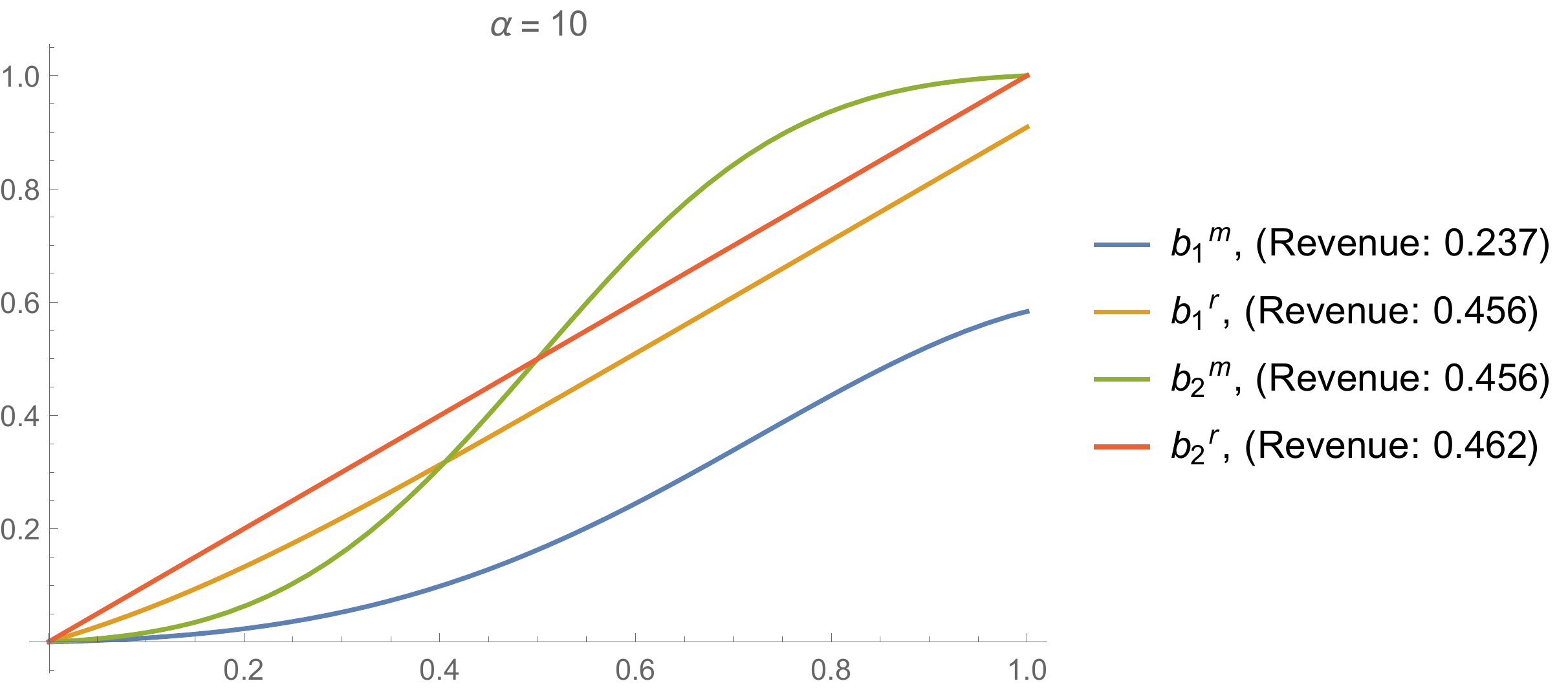} 
\par\end{centering}
\caption{Equilibrium bid functions for $\alpha\in\{1,5,10\}$}
\label{fig:Comparison_SPA_FPA} 
\end{figure}

To illustrate the role of correlation, the functions are shown for
$\alpha\in\{1,5,10\}$. Comparing FPA and SPA in the rational case,
we see the familiar revenue ranking that the SPA yields higher revenue
than the FPA with correlated types. This revenue ranking is preserved
in the case of misspecified bidders. Interestingly, with misspecified
bidders, the gap between SPA and FPA becomes more pronounced if values
are more correlated. This conforms well with the intuition for the
distortions in the bid function: In the SPA low types underbid and
high types overbid. In the FPA, the same forces lead the low types
to underbid. But this allows the higher types to shade their bids
more and the incentive to overbid does not compensate for this force.
This leads to much lower bids for misspecified types compared to the
rational equilibrium if the correlation is high.

Finally, we want to compare the efficiency of the SPA and FPA. This
comparison is not interesting in the purely rational or purely misspecified
cases since the symmetric equilibrium implies that both auction formats
are fully efficient. A comparison in the mixed case is challenging
because we are not able to compute the equilibrium in the FPA. To
make progress we consider the best response of a misspecified type
to the purely rational equilibrium. This allows us to show how efficiency
changes if we inject a small share of misspecified types in a rational
population. Figure \ref{fig:best_response_to_rational_eqm} shows
the resulting bidding strategies for $\alpha=1.5$.
\begin{figure}
\begin{centering}
\includegraphics[height=0.2\textheight]{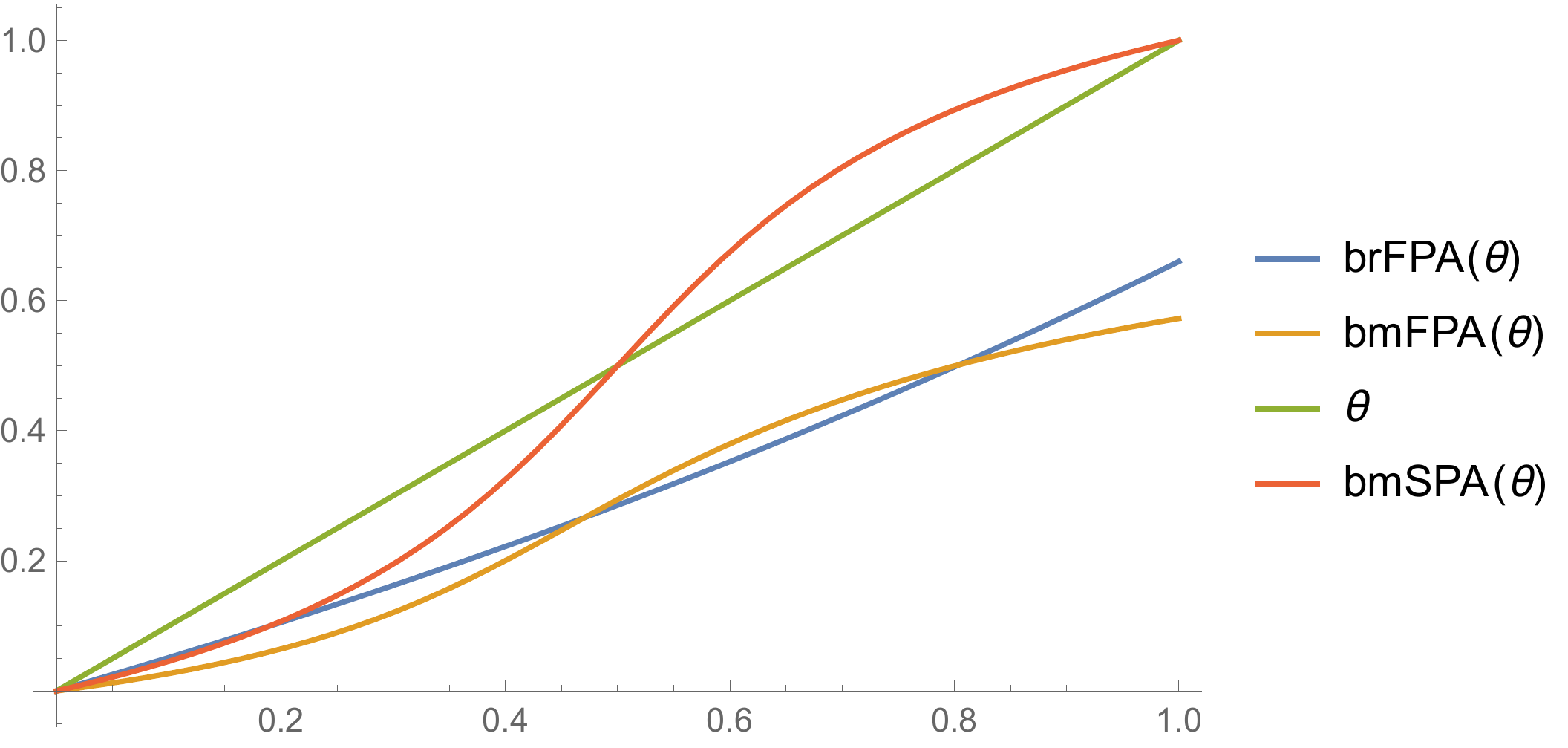} 
\par\end{centering}
\caption{For the FPA ($k=1$), and the SPA ($k=2$), $b_{k}^{r}(\theta_{i})$
is the rational strategy if $\lambda=1$; $b_{k}^{m}(\theta_{i})$
is the best response of misspecified type to data generated by the
purely rational equilibrium. ($\alpha=1.5$).}
\label{fig:best_response_to_rational_eqm} 
\end{figure}

To compare the efficiency we numerically compute how much efficiency
is lost in expectation if bidder one uses the purely rational strategy
and bidder two uses the misspecified response. This number gives the
rate at which efficiency decreases if we decrease $\lambda$ from
$\lambda=1$. In the example depicted in Figure \ref{fig:best_response_to_rational_eqm}
we have a marginal loss of $.0035$ for the FPA and $.0088$ for the
SPA. This means that the SPA is less efficient than the FPA.

\section{Auction-like Mechanisms\label{sec:Auction-like-Mechanisms}}

We now consider the possibility of implementing an efficient allocation
in the presence of both rational and misspecified buyers. We consider
a class of auction-like mechanisms, in which bidders can place a one-dimensional
bid $b\in B\subset\mathbb{R}$, and which allocate to the highest
bid (possibly adjusted by a bonus or malus). We assume that bidders
may choose not to participate in a mechanism in which case their utility
is zero.
\begin{defn}
An \emph{auction-like mechanism} is given by $M=\left[B,\left(W_{i}\right)_{i=1,2},\left(L_{i}\right)_{i=1,2},\phi_{1}\right]$.
$B=[\underline{b},\overline{b}]$ is the set of \emph{feasible bids}.
The \emph{allocation rule} $\phi_{1}:B\rightarrow B$ is a strictly
increasing function. The object is allocated to bidder 1 if $b_{1}>\phi_{1}(b_{2})$,
to bidder two if $b_{1}<\phi_{1}(b_{2})$, and with probability $1/2$
if $b_{1}=\phi_{1}(b_{2})$. We denote the inverse by $\phi_{2}=\phi_{1}^{-1}$.
The \emph{payment rules }are $W_{i}:B\times B\rightarrow\mathbb{R}_{0}^{+}$,
and $L_{i}:B\times B\rightarrow\mathbb{R}_{0}^{+}$, which specify
the payment bidder $i$ has to make as a function of the bids, if
she wins or loses, respectively. We assume that for each $i\in\{1,2\}$,
both functions $W_{i},L_{i}$ are weakly increasing in bidder $i$'s
own bid.\\
An auction-like mechanism is \emph{smooth} if for $i\in\{1,2\}$,
$\phi_{i}$, $W_{i}$, and $L_{i}$, are continuously differentiable
with derivatives that can be continuously extended to the boundary
of $B$. 
\end{defn}
The smoothness assumption is made for tractability. Many common auction
formats are smooth auction-like mechanisms. Our main result is that
if there are at least three possible ex-post valuations, then for
generic type distributions, no smooth auction-like mechanism exists
that has an efficient equilibrium.

To make this precise, we reformulate the types of agents. We denote
the \emph{interim valuation} of bidder $i$ with type $\theta_{i}$
by 
\[
w_{i}(\theta_{i}):=\mathbb{E}\left[v_{i}\middle|\theta_{i}\right].
\]
Given the normalization $0=v^{1}<\ldots<v^{K}=1$, we have $w_{i}(\theta_{i})\in[0,1]$.
For each $w_{i}\in[0,1]$, we denote the set of types $\theta_{i}$
that have interim valuation $w_{i}$ by 
\[
\Theta_{i}(w_{i}):=\left\{ \theta_{i}\in\Theta_{i}\middle|\mathbb{E}\left[v_{i}\middle|\theta_{i}\right]=w_{i}\right\} .
\]
For $w_{i}\in\{0,1\}$ this set is a singleton; and for all $w_{i}\in(0,1)$,
there exists a homeomorphism $x_{i}(\cdot;w_{i}):\Theta_{i}(w_{i})\rightarrow[0,1]^{K-2}$,
where $K=|V|$ is the number of ex-post valuations. We can therefore
write the type of bidder $i$ as $(w_{i},x_{i})\in[0,1]^{K-1}$. While
$w_{i}$ is the payoff-relevant part of the type, for $w_{i}\in(0,1),$
$x_{i}$ can be used to recover the belief $f(\theta_{j}|x_{i}^{-1}(x_{i};w_{i}))$
about bidder $j$'s type. Abusing notation we use $f(w_{1},x_{1},w_{2},x_{2})$
to denote the joint density of the buyers' types and assume that this
density is strictly positive.

Our main result is that for generic distributions, smooth auction-line
mechanisms do not have efficient equilibria. To state this formally,
we let $\mathcal{M}_{+}^{d}([0,1]^{2K-2})$ be the set of probability
measures on $[0,1]^{2K-2}$ that admit continuous and strictly positive
densities $f(w_{1},x_{1},w_{2},x_{2})$. We endow $\mathcal{M}_{+}^{d}([0,1]^{2K-2})$
with the uniform topology for densities. For given $V$ and $\lambda$,
let $\mathcal{I}(V,\lambda)\subset\mathcal{M}_{+}^{d}([0,1]^{2K-2})$
be the set of prior distributions for which all equilibria of any
smooth auction-like mechanism are inefficient.
\begin{thm}
\label{thm:Main_Result}Suppose $K=|V|\ge3$ and $\lambda\in(0,1)$.
Then for generic type distributions, there exists no smooth auction-like
mechanism with an efficient equilibrium. Formally, $\mathcal{I}(V,\lambda)$
is a residual subset of $\mathcal{M}_{+}^{d}([0,1]^{2K-2})$, that
is, it contains a countable intersection of open and dense subsets
of $\mathcal{M}_{+}^{d}([0,1]^{2K-2})$. 
\end{thm}
The notion of genericity used here is the same as in \citet{Gizatulina2017},
who show the genericity of full surplus extraction. The key step in
the proof is to show that in the presence of rational bidders, efficiency
requires that the mechanism is a second price auction. The reason
is that to achieve efficiency, the bid in an auction-like mechanism
must be a function of $w_{i}$ only. If there are more than two ex-post
valuations, for each $w_{i}\in(0,1)$, the set $\Theta(w_{i})$ is
a manifold of dimension $K-2\geq1$, and all types in $\Theta(w_{i})$
have identical interim expected valuations but different beliefs.
We show that for generic distributions, the requirement that the bid
is independent of the rational type's belief, implies that the mechanism
must be a second-price auction. We then finish the proof by extending
the result of Proposition \ref{prop:SPA_is_inefficient} to more than
two ex-post valuations (see Lemma \ref{lem:m-type_dont_bid_truthfully}
below), showing that in a second-price auction the misspecified type
does not bid truthfully, which rules out an efficient equilibrium.
\begin{rem}[Two ex-post valuations]
With only two ex-post valuations ($K=2$), our proof does not apply.
While the analysis of standard auctions in Section \ref{sec:Standard-Auctions}
suggests that bid functions of rational and misspecified types in
auction-like mechanisms differ, it is an open question whether auction-like
mechanism offer enough flexibility in choosing the payment rules so
that types of both sophistication can be incentivised to use an identical
bid function when $K=2$. 
\end{rem}
\begin{rem}[More than two bidders]
The restriction to two bidders has been made for simplicity. With
more than two bidders, we can consider misspecified types who have
access to data from past auctions with observations of the form $(b_{1},v_{1},\ldots,b_{N},v_{N})$,
where $N$ is the number of bidders. Such bidders will now rely on
$h(b_{-i}|v_{i})$, the pdf of $b_{-i}=(b_{j})_{j\neq i}$ conditional
on $v_{i}$, to form their beliefs about how variables of interest
are distributed. We can define auction-like mechanisms that award
the object to the highest bidder and specify payments as a function
of all bids. We conjecture that the key argument in our proof---namely
that efficiency requires the use of a second-price auction also works
with more than two bidders, as long as there are at least two ex-post
valuations. Moreover, an analogous result to Proposition \ref{prop:SPA_is_inefficient}
and Lemma \ref{lem:m-type_dont_bid_truthfully} implies that misspecified
types do not use the rational bid function in any equilibrium of the
second-price auction. 
\end{rem}

\subsection{Proof of Theorem \protect\ref{thm:Main_Result}}

\paragraph*{Regular Equilibria of Simple Mechanisms}

First, we show that it suffices to consider \emph{regular equilibria}
of \emph{simple mechanisms}. We call a smooth auction-like mechanism
\emph{simple} if it is of the form $M=\left[[0,1],\left(W_{i}\right),\left(L_{i}\right),Id\right]$,
where $\phi=Id$ denotes the identity so that the allocation rule
is symmetric. We call an equilibrium \emph{regular} if it is symmetric
and the bid of each generalized type $(w_{i},x_{i},s_{i})$ is given
by a continuous and strictly increasing function $b(w_{i})$ with
range $b([0,1])=[0,1]$. In other words, the bid only depends on the
interim valuations, but not on the identity, sophistication, or belief
$x_{i}$, of the bidder. Note that a regular equilibrium of a simple
mechanism is efficient. We denote the strictly increasing and continuous
inverse of $b(\cdot)$ by $\psi:[0,1]\rightarrow[0,1]$.
\begin{lem}
\label{lem:symmetric_mechanisms}Let $\tilde{M}=[\tilde{B},(\tilde{W}_{i}),(\tilde{L}_{i}),\tilde{\phi}_{1}]$
be a smooth auction-like mechanism with an efficient equilibrium $(\tilde{b}_{1}(w_{1},x_{1},s_{1})$,$\tilde{b}_{2}(w_{2},x_{2},s_{2}))$.
Then there exists a simple mechanism $M=\left[[0,1],\left(W_{i}\right),\left(L_{i}\right),Id\right]$,
with a regular (and hence efficient) equilibrium. 
\end{lem}
\begin{proof}
The proofs of all Lemmas can be found in the Appendix. 
\end{proof}
In light of Lemma \ref{lem:symmetric_mechanisms}, it suffices to
consider regular equilibria of simple mechanisms. The intuition behind
this result is that in an efficient mechanism with a symmetric allocation
rule,\footnote{Clearly, an mechanism with an asymmetric allocation rule can be made
symmetric by a simple monotonic transformation.} all bidders must use the same bids as function of their interim valuation.
The proof shows that mechanisms for which the bidding function has
discontinuities, these jumps can be removed in a way that preserves
the smoothness of the payment rules. Lemma \ref{lem:symmetric_mechanisms}
falls short of the revelation principle because the full revelation
argument may not preserve the smoothness of the payment rules if the
equilibrium of the original mechanism is non-smooth.

\paragraph*{Second-Price Auctions}

Next we derive a condition on the payment rules and equilibrium bid
function that characterizes regular equilibria of the second-price
auction. We denote the equilibrium difference in utility between winning
and losing of a bidder with bid $b=b(w_{i})$, whose bid is tied with
the opponent by 
\[
\delta_{i}(b)=\psi(b)-\left(W_{i}(b,b)-L_{i}(b,b)\right).
\]

In a regular equilibrium of the SPA, the rational type bids truthfully
($b(w)=w)$, and the payment rules satisfy $W_{i}(b,b)=b$ and $L_{i}\equiv0$,
so that $\delta_{i}(b)=0$ for all $b\in[0,1]$. The following Lemma
shows the converse result.
\begin{lem}
\label{lem:implications_delta=00003D00003D0}Consider a simple mechanism
$M=\left[[0,1],\left(W_{i}\right),\left(L_{i}\right),Id\right]$ with
a regular equilibrium. If $\delta_{i}(b)=0$, for $i\in\{1,2\}$ and
all $b\in[0,1]$, then $M$ is a second-price auction---that is,
for all $i\in\{1,2\}$, $L_{i}(b_{i},b_{j})=0$ for all $b_{i}\le b_{j}$
and $W_{i}(b_{i},b_{j}(w_{j}))=w_{j}$ whenever $b_{i}\ge b_{j}(w_{j})$. 
\end{lem}

\paragraph*{Differentiability of the Bidding Strategy}

To show that $\delta_{i}(b)=0$ for all bids we derive an implication
of $\delta_{i}(b)>0$ and show that it is violated generically. In
the derivations we will use first-order conditions and the following
Lemma shows that the inverse of the bid function, $\psi(b)$ is differentiable
if $\delta_{i}(b)>0$. The Lemma is based on the proof of Lemma 7
in \citet{Lizzeri2000}.
\begin{lem}
\label{lem:psi'>0}If $\delta_{i}(b_{0})>0$ for some $b_{0}\in[0,1]$,
then there exist a non-empty interval $(\alpha,\beta)\subset[0,1]$,
with $b_{0}\in(\alpha,\beta)$, such that $\psi$ is continuously
differentiable on $(\alpha,\beta)$, and $\psi^{\prime}(b)>0$ and
$\delta_{i}(b)>0$ for all $b\in(\alpha,\beta)$. 
\end{lem}

\paragraph*{For generic distributions efficiency requires $M=SPA$.}

Next, we show that $\delta_{i}(b)>0$ implies that a condition similar
to the full-surplus extraction condition \citep{McAfee1992} must
be violated, and prove results analog to \citet[henceforth GH17]{Gizatulina2017},
to show that for generic prior densities $f(w_{1},x_{1},w_{2},x_{2})$,
we must have $\delta_{i}(b)=0$ for all $b\in[0,1]$, $i\in\{1,2\}$,
and any regular equilibrium of a simple mechanism.

We begin by deriving an implication of $\delta_{i}(b)>0$. Fix $b\in(0,1)$
such that $\delta_{i}(b)>0$ and consider a rational bidder $i$ with
type $(w_{i},x_{i})$, where $w_{i}=\psi(b)$ and $x_{i}\in X$ is
arbitrary. In a regular equilibrium, this type maximizes (where we
use $j\neq i$ to denote the opponent): 
\[
\max_{b^{\prime}\in[0,1]}\int_{0}^{\psi(b^{\prime})}\left(\psi(b)-W_{i}(b^{\prime},b(w_{j}))\right)f(w_{j}|\psi(b),x_{i})dw_{j}-\int_{\psi(b^{\prime})}^{1}L_{i}(b^{\prime},b(w_{j}))f(w_{j}|\psi(b),x_{i})dw_{j}
\]
Given Lemma \ref{lem:psi'>0}, we can differentiate the objective
function with respect to $b^{\prime}$, and obtain the first-order
condition, which must hold for $b^{\prime}=b$: 
\begin{equation}
f(\tilde{w}_{j}=\psi(b)|\tilde{w}_{i}=\psi(b),x_{i})=\int_{0}^{1}\frac{\partial P_{i}(b,b(w_{j}))/\partial b_{i}}{\delta_{i}(b)\psi^{\prime}(b)}f(w_{j}|\tilde{w}_{i}=\psi(b),x_{i})dw_{j},\label{eq:FOC}
\end{equation}
where we simplify notation by denoting the payment of bidder $i$
as follows 
\[
P_{i}(b_{i},b_{j}):=\boldsymbol{1}_{\{b_{i}>b_{j}\}}W_{i}(b_{i},b_{j})+\boldsymbol{1}_{\{b_{i}<b_{j}\}}L_{i}(b_{i},b_{j}).
\]
Multiplying \eqref{eq:FOC} by $f(\tilde{w}_{i}=\psi(b),x_{i})/f(\tilde{w}_{i}=\tilde{w}_{j}=\psi(b))$,
and using 
\[
f(x_{i}|\tilde{w}_{i}=\tilde{w}_{j}=\psi(b))f(\tilde{w}_{i}=\tilde{w}_{j}=\psi(b))=f(x_{i},\tilde{w}_{i}=\tilde{w}_{j}=\psi(b)),
\]
we obtain for all $x_{i}\in X_{i}$: 
\begin{align}
f(x_{i}|\tilde{w}_{i} & =\tilde{w}_{j}=\psi(b))=\int_{0}^{1}m(b,\psi(b),w_{j})f(x_{i}|\tilde{w}_{i}=\psi(b),w_{j})dw_{j},\label{eq:FSE_condition}
\end{align}
where 
\[
m(b,\psi(b),w_{j})=\frac{\partial P_{i}(b,b(w_{j}))/\partial b_{i}f(\tilde{w}_{i}=\psi(b),w_{j})}{\delta_{i}(b)\psi^{\prime}(b)f(\tilde{w}_{i}=\tilde{w}_{j}=\psi(b))}.
\]
Since we consider a simple mechanism and prior densities $f\in\mathcal{M}_{+}^{d}([0,1]^{2K-2})$,
and $\psi^{\prime}(b)>0$, the term $m(b,\psi(b),w_{j})$ is finite
and non-negative. For fixed $b$, $m(b,\psi(b),\cdot)$ is in fact
a probability density on $[0,1]$.\footnote{Integrating both sides of \eqref{eq:FSE_condition} over $X$ we see
that $\int_{0}^{1}m(b,\psi(b),w_{j})dw_{j}=1$.}

Condition \eqref{eq:FSE_condition} states that the density $f(\cdot|\tilde{w}_{i}=\tilde{w}_{j}=\psi(b))$
can be expressed as a positive linear combination of the densities
$f(\cdot|\tilde{w}_{i}=\psi(b),w_{j})$ for $w_{j}\in[0,1]$, with
positive weights on $w_{j}\neq\psi(b)$. By virtually the same proof
as for Theorem 2.4 in GH17, we can show that for generic distributions
\eqref{eq:FSE_condition} is violated.

To state the result we need several definitions that mimic GH17. Let
$\mathcal{M}_{+}^{d}(X)$ be the set of absolutely continuous probabilities
measures on $X$ with strictly positive and continuous densities,
endowed with the topology induced by the $\sup$-norm for density
functions on $X$; let $\mathcal{C}([0,1],\mathcal{M}_{+}^{d}(X))$
be the set of continuous mappings from $[0,1]$ to $\mathcal{M}_{+}^{d}(X)$,
endowed with the topology of uniform convergence; and let $\mathcal{M}([0,1])$
be the set of probability measures on $[0,1]$, endowed with a topology
that is metrizable by a metric that is a convex function on $\mathcal{M}([0,1])\times\mathcal{M}([0,1])$.
Finally let $\mathcal{E}(w_{i})\subset\mathcal{C}([0,1],\mathcal{M}_{+}^{d}(X))$
be the set of continuous mappings that map $w\in[0,1]$ to densities
$g(\cdot|w)\in\mathcal{M}_{+}^{d}(X)$ that satisfy the following
condition: For all $\mu\in\mathcal{M}([0,1])$: 
\begin{equation}
g(x_{i}|w_{i})=\int_{0}^{1}g(x_{i}|w^{\prime})\mu(dw^{\prime}),\;\forall x_{i}\in X\qquad\implies\qquad\mu=\delta_{w_{i}}\label{eq:FSE_measure}
\end{equation}
where $\delta_{w_{i}}\in\mathcal{M}([0,1])$ is the Dirac measure
with a mass-point on $w_{i}$.
\begin{lem}[see Theorem 2.4 in \citealp{Gizatulina2017}]
\label{lem:genericity_conditional_density}For any $w_{i}\in(0,1)$,
the set $\mathcal{E}(w_{i})$ is a residual subset of $\mathcal{C}([0,1],\mathcal{M}_{+}^{d}(X))$,
that is, it is a countable intersection of open and dense subsets
of $\mathcal{C}([0,1],\mathcal{M}_{+}^{d}(X))$. 
\end{lem}
The implication of this Lemma is that for fixed $w_{i}\in(0,1)$,
and generic functions $w_{j}\mapsto f(\cdot|w_{i},w_{j})$ that map
$w_{j}$ to conditional densities $f(\cdot|w_{i},w_{j})$, any simple
mechanism with a regular equilibrium must satisfy $\delta_{i}(b(w_{i}))=0$.

This Lemma is insufficient for our purposes for two reasons. First,
we need to show that \emph{for generic priors}, the function that
maps $w_{j}$ to the conditional density $f(x_{i}|w_{i},w_{j})$ is
an element of $\mathcal{E}(w_{i})$, and second we need to show this
\emph{for all $w_{i}$}. To this end, for $i\in\{1,2\}$ let $\mathcal{W}_{i}$
be a countable and dense subset of $(0,1)$. We show that for generic
prior densities $f(w_{1},x_{1},w_{2},x_{2})$, the mapping that maps
$w_{j}\in[0,1]$ to the conditional density $f_{i}(\cdot|w_{i},w_{j})$
is an element of $\mathcal{E}_{i}(w_{i})$ for all $w_{i}\in\mathcal{W}_{i}$
and all $i\in\{1,2\}$. For the following Lemma, recall that $\mathcal{M}_{+}^{d}([0,1]^{2K-2})$
denotes the set of priors with strictly positive and continuous densities.
\begin{lem}[see Theorem 2.7 in \citealp{Gizatulina2017}]
\label{lem:genericity_prior}For $i\in\{1,2\}$, let $\mathcal{W}_{i}$
be a countable and dense subset of $(0,1)$. Let $\mathcal{F}$ be
the set of prior densities in $\mathcal{M}_{+}^{d}([0,1]^{2K-2})$
such that for all $i\in\{1,2\}$ and $w_{i}\in\mathcal{W}_{i}$, the
mapping $w_{j}\mapsto f(\cdot|w_{i},w_{j})$ is an element of $\mathcal{E}(w_{i})$.
Then $\mathcal{F}$ is a residual subset of $\mathcal{M}_{+}^{d}([0,1]^{2K})$,
that is it contains a countable intersection of open and dense subsets
of $\mathcal{M}_{+}^{d}([0,1]^{2K})$. 
\end{lem}
This Lemma implies that for generic prior densities $f(w_{1},x_{1},w_{2},x_{2})$,
any regular equilibrium of a simple mechanism must satisfy $\delta_{i}(b(w_{i}))=0$
for all $w_{i}\in\mathcal{W}_{i}$. Since the functions $b(\cdot)$
and $\delta_{i}(\cdot)$ are continuous and $\text{\ensuremath{\mathcal{W}_{i}}}$
is dense, this implies $\delta_{i}(b)=0$ for all $b\in[0,1]$. By
Lemma \ref{lem:implications_delta=00003D00003D0}, this implies that
for generic distributions, if a simple mechanism has a regular equilibrium,
then it must be the second-price auction.

\paragraph*{Bidding Strategy of the Misspecified Type in the Second-Price Auction}

So far we have made use of the rational type's first-order condition
to show that efficiency cannot be achieved with an auction-like mechanism
other than the SPA. To conclude the proof of Theorem \ref{thm:Main_Result}
we show that for generic distributions, misspecified types do not
use $b(w)=w$ in a SPA.
\begin{lem}
\label{lem:m-type_dont_bid_truthfully}Let $\lambda\in(0,1)$ and
suppose that $\mathbb{E}_{f}\left[\theta_{i}^{K}|w_{j}\le b\right]\neq\frac{\mathbb{E}_{f}\left[\theta_{i}^{K}\right]}{\mathbb{E}_{f}\left[\theta_{i}^{1}\right]}\mathbb{E}_{f}\left[\theta_{i}^{1}|w_{j}\le b\right]$
for some $i\in\{1,2\}$ and $b\in[0,1]$. In any equilibrium of the
second price auction where the rational types bid truthfully, some
types $(\theta_{i},m_{i})$ place a bid different from their interim
valuation. 
\end{lem}
It is easy to see that the subset of prior densities for which there
exists $i\in\{1,2\}$ and $b\in[0,1]$ such that $\mathbb{E}_{f}\left[\theta_{i}^{K}|w_{j}\le b\right]\neq\frac{\mathbb{E}_{f}\left[\theta_{i}^{K}\right]}{\mathbb{E}_{f}\left[\theta_{i}^{1}\right]}\mathbb{E}_{f}\left[\theta_{i}^{1}|w_{j}\le b\right]$
is open and dense $\mathcal{M}_{+}^{d}([0,1]^{2K})$ so that its intersection
with $\mathcal{F}$ is residual by Lemma \ref{lem:genericity_prior}.
This concludes the proof of Theorem \ref{thm:Main_Result}.

\section{Discussion and Extensions\label{sec:Discussion-and-Extensions}}

In this section we discuss alternative modeling approaches and extensions.
In Subsection \ref{subsec:Model-of-Belief}, we discuss alternative
models of belief formation from observed data when past bids and ex-post
valuations are observable. In Subsection \ref{subsec:Non-observability-of-losing}
we review possible approaches one could take for situations in which
data from past auctions only include the winning bid or the payment
made in past auctions. Finally, in Subsection \ref{subsec:General-Mechanisms}
we discuss general mechanisms that go beyond the auction-like class
considered in Section \ref{sec:Auction-like-Mechanisms}.

\subsection{Model of Belief Formation from Observed Data\label{subsec:Model-of-Belief}}

Our goal in this paper was to model belief formation of novice bidders
(the $m$-types) who lack a complete understanding of the auction
environment. Two basic assumptions have guided our modeling choices.
First, we have assumed that novice bidders are sophisticated in the
sense that they are able to use the empirical joint distribution of
observable variables to inform their own bid. Second, we have assumed
that novice bidders do not reason about how the bids of past bidders
were formed. In particular they do not form a conjecture or model
of the information available to past bidders and do not try to analyze
how such information drives observed behavior.

Due to missing data about signals (or types) of past bidders, novice
bidders are not able to learn the true joint distribution of signals/types,
ex-post valuations and opponent's bids. At the same time, a novice
bidder knows her own type $\theta_{i}$, and has access to the empirical
distribution of observable variables. She lacks knowledge how these
two should be combined, and a priori, many different ways of using
the data are conceivable, all of which rely on some implicit or explicit
assumptions. Following our second basic assumption, novice bidders
do not try to reason about how past bidders have determined their
bids. Instead they simply combine the joint distribution of observable
variables $v_{i}$ and $b_{j}$ with the belief about the distribution
of $v_{i}$ given by their type $\theta_{i}$ to evaluate the expected
payoff of different bids. This leads to a misspecified model in which
$v_{i}$ and $b_{j}$ are correlated even when conditioning on $\theta_{i}$.

We believe that this simple way of using the data is a plausible model
of an inexperienced bidder, but we do not want to claim that it is
the only way to think about data-driven belief formation. It is an
empirical question which models of belief formation are most suitable
and we hope that this article inspires experimental or empirical work
on this topic. In our model, other ways of forming beliefs may lead
to different misspecifications and deviations from rational behavior.
However, it is clear from the analysis in Section \ref{sec:Auction-like-Mechanisms},
that our main result is robust to the inclusion of various other forms
of misspecifications. In fact the presence of various types who differ
in their misspecification will make it harder to achieve efficiency
since a mechanism would have to account for all types, a task that
is already impossible in the presence for rational and a single misspecified
type.

One aspect in which bidders could differ is the \emph{precision} of
the available data. For example, past bids may only be available in
broad categories, such as ``high''\ or ``low''\ which would
indicate that the bid was above or below some threshold $b^{\ast}$.
A bidder with access to such aggregated data would have to make some
assumption about the non-identified distribution of the bid on the
intervals $[0,b^{\ast}]$ and $[b^{\ast},1]$. A simple starting point
could be the uniform distribution but bidders may also make a different
assumption. Given such an assumption, our notion of equilibrium can
be adapted. While an analysis is beyond the scope of this paper, we
conjecture that it would lead to similar conclusions as our present
model.

Alternatively, one could consider different degrees of sophistication
of novice bidders. We assume that novice bidders do not try to reason
about how past bidders determined their bids, nor about the possible
joint distribution of $v_{i},b_{j}$ and $\theta_{i}$. Novice bidders
are sophisticated in their ability to analyze data, but boundedly
rational in the sense that they do not question their method even
though there might be alternative ways of using past data to perform
bids.

Thinking about more sophisticated types, we may ask what additional
knowledge they would need to have in order to see that their model
is misspecified. In the data, one can see that $v_{i}$ and $b_{j}$
are independent conditional on $b_{i}$ since bids are a function
of $\theta_{i}$. However, without further assumptions how past bids
were formed, this does not allow to conclude that $v_{i}$ and $b_{j}$
are independent conditional on $\theta_{i}$. Hence, given the available
data, it is not obvious to an observer or the novice bidder, that
the $m$-type in our model uses a misspecified model.\footnote{The case of two possible ex-post valuations is special. Here, a more
sophisticated $m$-type might make the plausible assumptions that
(a) past bidders also had a one-dimensional type $\theta_{i}$ and
(b) bids are a strictly increasing function of $\theta_{i}$. Based
on these assumptions, the $m$-type could conclude from the data that
$v_{i}$ and $b_{j}$ are independent conditional on $\theta_{i}$,
leading her to behave like the rational type. Note however, that with
more than two possible ex-post valuations, the bidding strategy cannot
be injective, and therefore, without further assumptions about bidding
behavior, the $m$-type cannot conclude from the data that $v_{i}$
and $b_{j}$ are independent conditional on $\theta_{i}$.} A ``more sophisticated''\ type would therefore have to have access
to more data that would allow to validate structural assumptions about
past bids.

Conversely, we may think of less sophisticated types who have access
to the same data as our $m$-type but do not attempt or are unable
to use the statistical link between the $b_{j}$ and $v_{i}$. For
example this could be bidders who are not able to analyze large data
sets beyond producing marginal distributions of the opponents bids.
Alternatively, the bidder may know her expected valuation $E[v_{i}|\theta_{i}]$
but not the full distribution $\theta_{i}$ over ex-post valuations.
Such bidders may in some cases actually display less bias in their
bidding behavior since they do not use the statistical link between
the $b_{j}$ and $v_{i}$ that gives rise to a (perceived) conditional
correlation. For example, this is the case in second price auctions
in which such bidders would bid optimally, in contrast to the $m$-type
bidders we consider.

\subsection{Non-observability of losing bids and valuations\label{subsec:Non-observability-of-losing}}

In this paper we have followed the paradigm that data on past beliefs
is unobservable but bidders have full access to past bids and past
ex-post valuations. In practice, the ex-post valuations may not be
observed precisely, and perhaps only noisy signals of the true valuation
are available. To model such a situation, one could formulate the
type $\theta_{i}$ as a distribution over such signal realizations
and proceed as before.

More importantly, auctioneers may not disclose all bids in an auction.
In the following, we discuss possible approaches for the case that
only the \emph{winning bid} and the identity of the winner is observable
(as well as ex-post valuations of both bidders). Alternatively, the
auctioneer may disclose the identity of the winner and the \emph{payment}
she has to make. In first-price auction this is equivalent to disclosing
the winning bid but in an ascending auction or second-price auction,
the payment is equal to the second highest bid and the following discussion
has to be modified accordingly.

Consider the case that after an auction with $b_{i}>b_{j}$, the data
point $(i,b_{i},v_{i},v_{j})$ is observed, so that data about losing
bids is not available. In the present paper, we have taken $H(b_{j}|v_{i})$
to be the empirical distribution of the opponent's bid $b_{j}$ conditional
on valuation $v_{i}$ for bidder $i$. If losing bids are not observed,
this distribution is not directly accessible. We outline three exemplary
models of how a bidder may construct $H(b_{j}|v_{i})$ that reflect
different degrees of sophistication. For each approach, the constructed
$H(b_{j}|v_{i})$ can be plugged into our equilibrium framework and
the analysis would proceed as before.

A \emph{naive bidder} may ignore that the observations about opponent's
bids $b_{j}$ for a given valuation $v_{i}$ is selected and use the
(observable) distribution $H(b_{j}|v_{i},b_{j}>b_{i})$ instead of
$H(b_{j}|v_{i})$. This approach will lead bidder $i$ to think that
bidder $j$ bids higher than in reality, which creates additional
bias.\footnote{\citet{Jehiel2018} uses a similar selection neglect to demonstrate
how investor overoptimism can arise if investors only observe realized
past projects. Note however, that in the auction context, the direction
of the bias crucially depends on what is disclosed. Indeed, if the
second highest bid is disclosed, a naive bidder may think that bidder
$j$ bids lower than in reality.}

A \emph{semi-naive bidder} may be aware that for each $v_{i}$ she
only observes a selected sample of opponent's bids $b_{j}$ which
satisfy $b_{j}>b_{i}$. For all other observations with a given $v_{i}$,
$b_{i}$ is known and she can only infer that $b_{j}<b_{i}$. The
bidder could then attempt to complete the missing data by assuming
some distribution $\tilde{H}(b_{j}|v_{i},b_{j}<b_{i})$. A natural
starting point would be the uniform distribution. We call this bidder
semi-naive since she makes some ad hoc assumption about $\tilde{H}(b_{j}|v_{i},b_{j}<b_{i})$,
but at least she makes an attempt to correct for the selected sample.
Given this approach, one could construct a distribution $H(b_{j}|v_{i})$
that combines the empirical distribution $H(b_{j}|v_{i},b_{j}>b_{i})$
and the assumed distributions $\tilde{H}(b_{j}|v_{i},b_{j}<b_{i})$.

Finally, a \emph{sophisticated bidder} may attempt to estimate the
distribution of $b_{j}$ conditional on $v_{i}$, using some structural
model. Since the correlation between $b_{1}$ and $b_{2}$ cannot
be assessed from the data, a natural starting point is that a bidder
takes them to be independent (conditional on $v_{1}$ and $v_{2}$),
and tries to identify the marginal conditional distribution $H(b_{j}|v_{i})$
from the data. An identical identification problem arises in competing
risk models. Translated into our context, the results of \citet{Tsiatis1975}
show that for any (not necessarily independent) joint distribution
of the bids $b_{1}$ and $b_{2}$, one can construct unique marginal
distributions that, under the assumption of independence are consistent
with the observed data. The independence assumption is thus not testable
and the sophisticated bidder is always able to pursue her approach.

Common to all three approaches is that, the naive, semi-naive, and
the sophisticated bidder will deviate from the rational bid in the
second-price auction if $\theta_{1}$ and $\theta_{2}$ are not independent.\footnote{Interestingly, only the last approach will have the converse property
that when the distributions of types are independent, bidders are
behaving optimally. In this sense, it is the closest to the insights
developed in this paper.} This is the case since in all approaches the bidder believes that
conditional on $\theta_{i}$, $v_{i}$ and $b_{j}$ are correlated.
Therefore, the impossibility of an efficient auction-like mechanism
continues to hold since the presence of the rational type requires
the use of the second-price auction, and as before, the $m$-type
does not bid truthfully in a second-price auction.

\subsection{General Mechanisms\label{subsec:General-Mechanisms}}

We have focused on a class of auction-like mechanisms in which bids
are one-dimensional and the highest bid wins. This is a natural class
that covers many practically relevant auction formats. Nevertheless,
it begs the question whether the designer could achieve efficiency
with more elaborate mechanisms. The literature on full surplus extraction
with correlated types has demonstrated that under conditions that
hold generically in the standard model, (1) the surplus of an incentive
compatible allocation rule can be fully extracted (almost fully with
continuous types), and (2) with discrete types, any allocation rule
can be implemented in a Bayes-Nash equilibrium \citep{Cremer1988,McAfee1992,Gizatulina2017}.
We are interested in the implementation of an allocation rule that
is not known to be incentive compatible given the unusual nature (the
cognitive dimension) of a type in our setting. Therefore, a generalization
of the second result for continuous types and a mix of rational and
misspecified bidders is needed. GH17 show that for generic distributions,
the belief $f(\theta_{j}|\theta_{i})$ for given $\theta_{i}$ cannot
be expressed as a convex combination of other types beliefs. Moreover,
in an incentive compatible direct mechanism, the belief of an $m$-type
about the messages of opponents is given by\footnote{If the designer can use general mechanisms, we could consider direct
mechanisms in which the message space coincides with the generalized
type-space $M=T$. Any equilibrium of a general mechanism can be replicated
by a truthful equilibrium of a direct mechanism. The truthful direct
mechanism is a tool in the analysis that should not be taken literally
if bidder's are not aware of their sophistication. Therefore if agents
are not aware of their sophistication, it is in general not obvious
how to replicate the equilibrium of a direct mechanism, using an indirect
mechanism in which agents are not asked for their sophistication directly.
However, if we want to replicate a direct mechanism that uses scoring
rules as suggested below, an indirect implementation would elicit
beliefs from the agents. This does not require agents to report their
sophistication.} 
\begin{equation}
h(\theta_{j}|\theta_{i})=\sum_{k=1}^{K}\theta_{i}^{k}f(\theta_{j}|v_{i}^{k})=\sum_{k=1}^{K}\theta_{i}^{k}f(\theta_{j}|\theta_{i}=e_{k}).\label{eq:m-type-belief function}
\end{equation}
where $e_{k}^{\ell}=\boldsymbol{1}_{k=\ell}$. Together this implies
that no two generalized types $(\theta_{i},s_{i})$ and $(\theta_{i}^{\prime},s_{i}^{\prime})$
have identical beliefs. Using a strictly proper scoring rule of the
type considered in \citet{Johnson1990}, we can therefore construct
payment rules that (a) require an expected payment equal to zero from
an agent who truthfully reports his belief, and (b) lead to a strictly
positive expected payment from any misreport. We conjecture that combining
such a payment rule with the efficient allocation rule allows to deter
any non-local deviations. However, it is an open question if there
exists a payment rule that (virtually) implements the efficient allocation
rule (or any other allocation rule that is not known to be incentive
compatible). As far as we know, an analog result has not been shown
with a continuous type space even in the standard model, and exploring
this direction is beyond the scope of this paper.

At the same time, even if a positive result can be achieved using
scoring rules, the resulting mechanism will be highly unrealistic
and schemes like this have been criticized even in the standard model
since they rely on detailed knowledge of the environment by the designer.
In our context, the designer would not only rely on detailed knowledge
about the type distribution, she would also have to know exactly how
misspecified types form beliefs based on past data. If, by contrast,
there is a rich set of possible misspecifications that arise from
different ways in which agents may use the available data, it seems
unlikely that a designer has the detailed knowledge required to design
a suitable incentive scheme. Moreover, for sufficiently rich sets
of misspecification specifications, it may not even be possible to
separate all types. For example, if there are many different types
who observe past data with different coarseness, then the result that
no two different types have the same beliefs about some aspects of
opponent's type is likely violated.

\newpage{}

\appendix

\section{Omitted Proofs\label{sec:Omitted-Proofs}}

\subsection{Proof of Proposition \protect\ref{prop:SPA_is_inefficient}}
\begin{proof}[Proof of Proposition \protect\ref{prop:SPA_is_inefficient}]
If $\lambda\in(0,1)$, efficiency would require that $b^{m}(\theta_{i})=\theta_{i}$
which implies 
\begin{align*}
H^{\text{SPA}}(b\mid v_{i}=1) & =\int_{0}^{1}F(b\mid\tilde{\theta}_{i})\tilde{\theta}_{i}\frac{f(\tilde{\theta}_{i})}{\mathbb{E}[\tilde{\theta}_{i}]}d\tilde{\theta}_{i},\\
H^{\text{SPA}}(b\mid v_{i}=0) & =\int_{0}^{1}F(b\mid\tilde{\theta}_{i})(1-\tilde{\theta}_{i})\frac{f(\tilde{\theta}_{i})}{\mathbb{E}[1-\tilde{\theta}_{i}]}d\tilde{\theta}_{i}.
\end{align*}
Moreover, we must have 
\[
\theta_{i}\in\arg\max_{b}\left\{ \theta_{i}H^{\text{SPA}}(b\mid v_{i}=1)-\theta_{i}\int_{0}^{b}xdH^{\text{SPA}}(x|v_{i}=1)-(1-\theta_{i})\int_{0}^{b}xdH^{\text{SPA}}(x|v_{i}=0)\right\} 
\]
Differentiating the objective function and setting $b=\theta_{i}$
yields 
\[
(1-\theta_{i})\theta_{i}\left[H^{\text{SPA}\prime}(\theta_{i}\mid v_{i}=1)-H^{\text{SPA}\prime}(\theta_{i}\mid v_{i}=0)\right]
\]
We have 
\begin{align*}
 & H^{\text{SPA}\prime}(\theta_{i}\mid v_{i}=1)-H^{\text{SPA}\prime}(\theta_{i}\mid v_{i}=0)\\
= & \int_{0}^{1}f(\theta_{i}\mid\tilde{\theta}_{i})\tilde{\theta}_{i}\frac{f(\tilde{\theta}_{i})}{\mathbb{E}[\tilde{\theta}_{i}]}d\tilde{\theta}_{i}-\int_{0}^{1}f(\theta_{i}\mid\tilde{\theta}_{i})(1-\tilde{\theta}_{i})\frac{f(\tilde{\theta}_{i})}{\mathbb{E}[1-\tilde{\theta}_{i}]}d\tilde{\theta}_{i}\\
= & \int_{0}^{1}\left[\frac{\tilde{\theta}_{i}}{\mathbb{E}[\tilde{\theta}_{i}]}-\frac{1-\tilde{\theta}_{i}}{1-\mathbb{E}[\tilde{\theta}_{i}]}\right]f(\theta_{i}\mid\tilde{\theta}_{i})f(\tilde{\theta}_{i})d\tilde{\theta}_{i}\\
= & f(\theta_{i})\int_{0}^{1}\left[\frac{\tilde{\theta}_{i}}{\mathbb{E}[\tilde{\theta}_{i}]}-\frac{1-\tilde{\theta}_{i}}{1-\mathbb{E}[\tilde{\theta}_{i}]}\right]f(\tilde{\theta}_{i}|\theta_{i})d\tilde{\theta}_{i}\\
= & f(\theta_{i})\left[\frac{\mathbb{E}[\tilde{\theta}_{i}|\theta_{j}=\theta_{i}]}{\mathbb{E}[\tilde{\theta}_{i}]}-\frac{1-\mathbb{E}[\tilde{\theta}_{i}|\theta_{j}=\theta_{i}]}{1-\mathbb{E}[\tilde{\theta}_{i}]}\right]
\end{align*}
Hence, for bidding $\theta_{i}$ to be optimal for the misspecified
type we must have for all $\theta_{i}$: 
\begin{align*}
\frac{\mathbb{E}[\tilde{\theta}_{i}|\theta_{j}=\theta_{i}]}{\mathbb{E}[\tilde{\theta}_{i}]}-\frac{1-\mathbb{E}[\tilde{\theta}_{i}|\theta_{j}=\theta_{i}]}{1-\mathbb{E}[\tilde{\theta}_{i}]} & =0\\
\iff\mathbb{E}[\tilde{\theta}_{i}|\theta_{j}=\theta_{i}] & =\mathbb{E}[\tilde{\theta}_{i}]
\end{align*}
If the last line holds for all $\theta_{i}$ we must have 
\begin{align*}
\int_{0}^{1}\tilde{\theta}_{i}f(\tilde{\theta}_{i}|\theta_{j})d\tilde{\theta}_{i} & =\mathbb{E}[\tilde{\theta}_{i}],\qquad\forall\theta_{j},\\
\iff\qquad\int_{0}^{1}\tilde{\theta}_{i}\theta_{j}f(\tilde{\theta}_{i},\theta_{j})d\tilde{\theta}_{i} & =\mathbb{E}[\tilde{\theta}_{1}]\theta_{j}f(\theta_{j}),\qquad\forall\theta_{j},\\
\implies\qquad E\left[\tilde{\theta}_{i}\theta_{j}\right] & =\left(\mathbb{E}[\tilde{\theta}_{i}]\right)^{2}.
\end{align*}
The last line implies that we must have $\text{Corr}[\theta_{1},\theta_{2}]=0$
if the misspecified types first-order condition is satisfied for $b=\theta_{i}$
for all $\theta_{i}$. Therefore, if $\text{Corr}[\theta_{1},\theta_{2}]\neq0$,
there are types for which a misspecified bidder will not bid $\theta_{i}$
and since $b^{r}(\theta_{j})=\theta_{j}$ for all types and $\lambda\in(0,1)$,
the allocation will be inefficiency for some type profiles. 
\end{proof}

\subsection{Proof of Proposition \protect\ref{prop:FPA_is_inefficient}}
\begin{proof}[Proof of Proposition \protect\ref{prop:FPA_is_inefficient}]
An efficient allocation requires that $b^{r}(\theta_{i})=b^{m}(\theta_{i})=b(\theta_{i})$
for all $\theta_{i}\in[0,1]$. We denote the inverse of $b(\cdot)$
by $\theta(Lemma[lem:implications_{w}=W(b,b)]thenimplie\cdot)$.

The rational type's bid solves 
\[
\max_{b}\left(\theta_{i}-b\right)F(\theta(b)|\theta_{i})
\]
The FOC yields 
\begin{align}
-F(\theta_{i}|\theta_{i})+\left(\theta_{i}-b(\theta_{i})\right)f(\theta_{i}|\theta_{i})\theta_{i}^{\prime}(b(\theta_{i})) & =0\nonumber \\
\iff\qquad b^{\prime}(\theta_{i}) & =(\theta_{i}-b(\theta_{i}))\frac{f(\theta_{i}|\theta_{i})}{F(\theta_{i}|\theta_{i})}.\label{eq:ode_eff_fpa_rational}
\end{align}
The solution with boundary condition $b(0)=0$ is 
\[
b(\theta_{i})=\int_{0}^{\theta_{i}}xe^{-\int_{x}^{\theta_{i}}\frac{f(y|y)}{F(y|y)}dy}\frac{f(x|x)}{F(x|x)}dx.
\]

The misspecified type maximizes \eqref{eq:m-problem_FPA} 
\[
\max_{b}\left(1-b\right)\theta_{i}H^{\text{FPA}}(b\mid v_{i}=1)-b(1-\theta_{i})H^{\text{FPA}}(b|v_{i}=0).
\]
with 
\begin{align*}
H^{\text{FPA}}(b\mid v_{i}=1) & =\int_{0}^{1}F(\theta(b)\mid\tilde{\theta}_{i})\tilde{\theta}_{i}\frac{f(\tilde{\theta}_{i})}{E[\tilde{\theta}_{i}]}d\tilde{\theta}_{i},\\
H^{\text{FPA}}(b\mid v_{i}=0) & =\int_{0}^{1}F(\theta(b)\mid\tilde{\theta}_{i})(1-\tilde{\theta}_{i})\frac{f(\tilde{\theta}_{i})}{E[1-\tilde{\theta}_{i}]}d\tilde{\theta}_{i}.
\end{align*}
This yields 
\begin{align*}
 & \theta_{i}H^{\text{FPA}}(b(\theta_{i})\mid v_{i}=1)+(1-\theta_{i})H^{\text{FPA}}(b(\theta_{i})|v_{i}=0)\\
 & =\left(1-b(\theta_{i})\right)\theta_{i}H^{\text{FPA}\prime}(b(\theta_{i})\mid v_{i}=1)-b(\theta_{i})(1-\theta_{i})H^{\text{FPA}\prime}(b(\theta_{i})|v_{i}=0)
\end{align*}
Using 
\begin{align*}
H^{\text{FPA}\prime}(b\mid v_{i}=1) & =\theta^{\prime}(b)\int_{0}^{1}f(\theta(b)\mid\tilde{\theta}_{i})\tilde{\theta}_{i}\frac{f(\tilde{\theta}_{i})}{E[\tilde{\theta}_{i}]}d\tilde{\theta}_{i}=\theta^{\prime}(b)\frac{E[\tilde{\theta}_{i}|\theta(b)]}{E[\tilde{\theta}_{i}]}f(\theta(b))\\
H^{\text{FPA}\prime}(b\mid v_{i}=0) & =\theta^{\prime}(b)\int_{0}^{1}f(\theta(b)\mid\tilde{\theta}_{i})(1-\tilde{\theta}_{i})\frac{f(\tilde{\theta}_{i})}{E[1-\tilde{\theta}_{i}]}d\tilde{\theta}_{i}=\theta^{\prime}(b)\frac{1-E[\tilde{\theta}_{i}|\theta(b)]}{1-E[\tilde{\theta}_{i}]}f(\theta(b))
\end{align*}
we have 
\begin{align*}
 & \theta_{i}\int_{0}^{1}F(\theta_{i}\mid\tilde{\theta}_{i})\tilde{\theta}_{i}\frac{f(\tilde{\theta}_{i})}{E[\tilde{\theta}_{i}]}d\tilde{\theta}_{i}+(1-\theta_{i})\int_{0}^{1}F(\theta_{i}\mid\tilde{\theta}_{i})(1-\tilde{\theta}_{i})\frac{f(\tilde{\theta}_{i})}{E[1-\tilde{\theta}_{i}]}d\tilde{\theta}_{i}\\
= & \left(1-b(\theta_{i})\right)\theta_{i}\theta^{\prime}(b(\theta_{i}))\frac{E[\tilde{\theta}_{i}|\theta_{i}]}{E[\tilde{\theta}_{i}]}f(\theta_{i})-b(\theta_{i})(1-\theta_{i})\theta^{\prime}(b(\theta_{i}))\frac{1-E[\tilde{\theta}_{i}|\theta_{i}]}{1-E[\tilde{\theta}_{i}]}f(\theta_{i})\\
\iff\qquad\qquad & \theta_{i}\int_{0}^{1}F(\theta_{i}\mid\tilde{\theta}_{i})\tilde{\theta}_{i}\frac{f(\tilde{\theta}_{i})}{E[\tilde{\theta}_{i}]}d\tilde{\theta}_{i}+(1-\theta_{i})\int_{0}^{1}F(\theta_{i}\mid\tilde{\theta}_{i})(1-\tilde{\theta}_{i})\frac{f(\tilde{\theta}_{i})}{E[1-\tilde{\theta}_{i}]}d\tilde{\theta}_{i}\\
= & \frac{1-b(\theta_{i})}{b^{\prime}(\theta_{i})}\theta_{i}\frac{E[\tilde{\theta}_{i}|\theta_{i}]}{E[\tilde{\theta}_{i}]}f(\theta_{i})-\frac{b(\theta_{i})}{b^{\prime}(\theta_{i})}(1-\theta_{i})\frac{1-E[\tilde{\theta}_{i}|\theta_{i}]}{1-E[\tilde{\theta}_{i}]}f(\theta_{i})\\
\iff b^{\prime}(\theta_{i})= & \frac{\left(1-b(\theta_{i})\right)\theta_{i}\frac{E[\tilde{\theta}_{i}|\theta_{i}]}{E[\tilde{\theta}_{i}]}f(\theta_{i})-b(\theta_{i})(1-\theta_{i})\frac{1-E[\tilde{\theta}_{i}|\theta_{i}]}{1-E[\tilde{\theta}_{i}]}f(\theta_{i})}{\theta_{i}\int_{0}^{1}F(\theta_{i}\mid\tilde{\theta}_{i})\tilde{\theta}_{i}\frac{f(\tilde{\theta}_{i})}{E[\tilde{\theta}_{i}]}d\tilde{\theta}_{i}+(1-\theta_{i})\int_{0}^{1}F(\theta_{i}\mid\tilde{\theta}_{i})(1-\tilde{\theta}_{i})\frac{f(\tilde{\theta}_{i})}{E[1-\tilde{\theta}_{i}]}d\tilde{\theta}_{i}}\\
= & \theta_{i}\frac{\frac{E[\tilde{\theta}_{i}|\theta_{i}]}{E[\tilde{\theta}_{i}]}f(\theta_{i})}{\theta_{i}\int_{0}^{1}F(\theta_{i}\mid\tilde{\theta}_{i})\tilde{\theta}_{i}\frac{f(\tilde{\theta}_{i})}{E[\tilde{\theta}_{i}]}d\tilde{\theta}_{i}+(1-\theta_{i})\int_{0}^{1}F(\theta_{i}\mid\tilde{\theta}_{i})(1-\tilde{\theta}_{i})\frac{f(\tilde{\theta}_{i})}{E[1-\tilde{\theta}_{i}]}d\tilde{\theta}_{i}}\\
 & -b(\theta_{i})\frac{\theta_{i}\frac{E[\tilde{\theta}_{i}|\theta_{i}]}{E[\tilde{\theta}_{i}]}f(\theta_{i})-(1-\theta_{i})\frac{1-E[\tilde{\theta}_{i}|\theta_{i}]}{1-E[\tilde{\theta}_{i}]}f(\theta_{i})}{\theta_{i}\int_{0}^{1}F(\theta_{i}\mid\tilde{\theta}_{i})\tilde{\theta}_{i}\frac{f(\tilde{\theta}_{i})}{E[\tilde{\theta}_{i}]}d\tilde{\theta}_{i}+(1-\theta_{i})\int_{0}^{1}F(\theta_{i}\mid\tilde{\theta}_{i})(1-\tilde{\theta}_{i})\frac{f(\tilde{\theta}_{i})}{E[1-\tilde{\theta}_{i}]}d\tilde{\theta}_{i}}\\
= & (\theta_{i}-b(\theta_{i}))\frac{f(\theta_{i}|\theta_{i})}{F(\theta_{i}|\theta_{i})}
\end{align*}
Where the last line follows from \eqref{eq:ode_eff_fpa_rational}.
Matching coefficients, we get 
\[
\frac{\frac{E[\tilde{\theta}_{i}|\theta_{i}]}{E[\tilde{\theta}_{i}]}f(\theta_{i})}{\theta_{i}\int_{0}^{1}F(\theta_{i}\mid\tilde{\theta}_{i})\tilde{\theta}_{i}\frac{f(\tilde{\theta}_{i})}{E[\tilde{\theta}_{i}]}d\tilde{\theta}_{i}+(1-\theta_{i})\int_{0}^{1}F(\theta_{i}\mid\tilde{\theta}_{i})(1-\tilde{\theta}_{i})\frac{f(\tilde{\theta}_{i})}{E[1-\tilde{\theta}_{i}]}d\tilde{\theta}_{i}}=\frac{f(\theta_{i}|\theta_{i})}{F(\theta_{i}|\theta_{i})}
\]
and 
\[
\frac{\theta_{i}\frac{E[\tilde{\theta}_{i}|\theta_{i}]}{E[\tilde{\theta}_{i}]}f(\theta_{i})-(1-\theta_{i})\frac{1-E[\tilde{\theta}_{i}|\theta_{i}]}{1-E[\tilde{\theta}_{i}]}f(\theta_{i})}{\theta_{i}\int_{0}^{1}F(\theta_{i}\mid\tilde{\theta}_{i})\tilde{\theta}_{i}\frac{f(\tilde{\theta}_{i})}{E[\tilde{\theta}_{i}]}d\tilde{\theta}_{i}+(1-\theta_{i})\int_{0}^{1}F(\theta_{i}\mid\tilde{\theta}_{i})(1-\tilde{\theta}_{i})\frac{f(\tilde{\theta}_{i})}{E[1-\tilde{\theta}_{i}]}d\tilde{\theta}_{i}}=\frac{f(\theta_{i}|\theta_{i})}{F(\theta_{i}|\theta_{i})}
\]
Combining these we have 
\begin{align*}
\frac{E[\tilde{\theta}_{i}|\theta_{i}]}{E[\tilde{\theta}_{i}]} & =\theta_{i}\frac{E[\tilde{\theta}_{i}|\theta_{i}]}{E[\tilde{\theta}_{i}]}-(1-\theta_{i})\frac{1-E[\tilde{\theta}_{i}|\theta_{i}]}{1-E[\tilde{\theta}_{i}]}\\
\frac{E[\tilde{\theta}_{i}|\theta_{i}]}{E[\tilde{\theta}_{i}]} & =\frac{1-E[\tilde{\theta}_{i}|\theta_{i}]}{1-E[\tilde{\theta}_{i}]}
\end{align*}
This is the same condition as for the SPA which requires that $\text{Corr}[\theta_{1},\theta_{2}]=0$. 
\end{proof}

\subsection{Proof of Lemma \protect\ref{lem:symmetric_mechanisms}}
\begin{proof}[Proof of Lemma \ref{lem:symmetric_mechanisms}]
Consider the equilibrium of the original mechanism $\tilde{M}.$
For each bidder $i$ and each $s_{i}\in\{r,m\}$, we define a (non-empty)
correspondence that contains all bids that types with expected valuation
$w_{i}$ use. 
\[
b_{i}^{s_{i}}(w_{i})=\tilde{b}_{i}(w_{i},X,s_{i})
\]
where $X=[0,1]^{K-2}.$ We prove the lemma in three steps: (1) we
obtain an efficient equilibrium of the original mechanism with single-valued
correspondences (or functions) $\hat{b}_{i}^{s_{i}}$. (2) We show
that these functions satisfy $\hat{b}_{i}^{r}(w)=\hat{b}_{i}^{s}(w)=\tilde{\phi}_{i}(\hat{b}_{j}^{r}(w))=\tilde{\phi}_{i}(\hat{b}_{j}^{r}(w))$,
and a change of variable allows us to construct a mechanism $\check{M}=\left(B,\left(\check{W}_{i}\right),\left(\check{L}_{i}\right),Id\right)$
that has an efficient equilibrium in which $\check{b}_{i}^{r}(w)=\check{b}_{i}(w)=\check{b}_{j}^{r}(w)=\check{b}_{j}^{r}(w)=\check{b}(w)$.
(3) We remove jump continuities in $\check{b}(w)$ and normalize the
range of $\check{b}(w)$ to obtain a mechanism $M=\left([0,1],\left(W_{i}\right),\left(L_{i}\right),Id\right)$
so that the (normalized) continuous part of $\check{b}(w)$ is an
efficient equilibrium. We show that removing the discontinuities does
not destroy the smoothness of the simple mechanism $M$.

\textbf{Step 1:} First, note that efficiency requires that the correspondences
$b_{i}^{s_{i}}$ for $i\in\{1,2\}$ must be strictly increasing, meaning
any selection must be strictly increasing. We denote the point-wise
infimum and supremum of the correspondence by $\bl_{i}^{s_{i}}(w)=\inf b_{i}^{s_{i}}(w_{i})$
and $\bh_{i}^{s_{i}}(w)=\inf b_{i}^{s_{i}}(w_{i})$. Note that the
infimum $\bl_{i}^{s_{i}}(w)$ is strictly increasing if any selection
from $b_{i}^{s_{i}}(w)$ is strictly increasing.

Suppose for some $w_{i}$, $b_{i}^{r}(w_{i})$ is not single-valued.
Efficiency and the fact that the in requires that for every $b_{i}\in[\bl_{i}^{s_{i}}(w),\bh_{i}^{s_{i}}(w)]$,
$\left(b_{j}^{s_{j}}\right)^{-1}(\phi_{2}(b_{i}))\subset\{w_{i}\}$,
that is, any bid in the closed interval between the between the infimal
and supremal bid that bidder $i$ with interim value $w_{i}$ places
in equilibrium is either not placed by bidder $j$ or it is placed
by a bidder with the same interim value. We can include the infimum
(and supremum) since $w_{j}\in\left(b_{j}^{s_{j}}\right)^{-1}(\phi_{2}(\bl_{i}^{s_{i}}(w)))$
for some $w_{j}<w_{i}$ would imply that there exists $w_{i}'\in(w_{j},w_{i})$
such that $b_{i}'<\bl_{i}^{s_{i}}(w)$ for some $b'_{i}\in b_{i}^{s_{i}}(w_{i}')$,
which violates efficiency.

Since the probability that $w_{j}=w_{i}$ conditional on $(w_{i},x_{i})$
is zero for all $x_{i}\in X_{i}$, the rational type is indifferent
between all bids in $[\bl_{i}^{s_{i}}(w),\bh_{i}^{s_{i}}(w)]$. We
set $\hat{b}_{i}(w_{i},x_{i},r):=\hat{b}_{i}^{r}(w_{i}):=\bl_{i}^{r}(w)$.
Similar steps show that we can set $\hat{b}_{i}(w_{i},x_{i},m):=\hat{b}_{i}^{m}(w_{i}):=\bl_{i}^{m}(w)$.

Since the probability that $E[v_{i}|\theta_{i}]=w_{i}$ is zero, and
there are at most countably many discontinuities, this modification
of $\tilde{b}_{i}$ to $\hat{b}_{i}$ does not change the incentives
of bidder $j$ so that we have constructed a new equilibrium in which
the correspondences of bidder $i$ are single valued. We can apply
the same modification to the strategy of bidder $j$. Clearly these
modification preserve efficiency since $b_{j}^{s_{j}}(w_{j})<\phi_{2}(\inf\tilde{b}_{i}^{r}(w_{i}))$
whenever $w_{j}<w_{i}$.

\textbf{Step 2:} We have shown in Step 1 that there exists an efficient
equilibrium of $\tilde{M}$ that is given by the function $\hat{b}_{i}^{s}(w)$,
$i\in\{1,2\}$, $s\in\{r,m\}$. Clearly, efficiency requires that
$\hat{b}_{i}^{r}(w)=\hat{b}_{i}^{m}(w)=\phi_{i}(\hat{b}_{j}^{r}(w))=\phi_{i}(\hat{b}_{j}^{m}(w))=:\hat{b}_{i}(w)$
for almost every $w$. The only exceptions are a countable set of
interim values where all functions have a jump-discontinuity. Here
we can redefine $\hat{b}_{i}^{r}(w)=\hat{b}_{i}^{m}(w)=\hat{b}_{i}(w):=\lim_{w'\uparrow w}\min\left\{ \hat{b}_{i}^{r}(w'),\hat{b}_{i}^{m}(w'),\phi_{i}(\hat{b}_{j}^{r}(w')),\phi_{i}(\hat{b}_{j}^{m}(w'))\right\} $
for $i\neq j$, so that $\hat{b}_{i}^{r}(w)=\hat{b}_{i}^{m}(w)=\phi_{i}(\hat{b}_{j}^{r}(w))=\phi_{i}(\hat{b}_{j}^{m}(w))=\hat{b}_{i}(w)$
for every $w$, and $\hat{b}_{i}(w)$ is left-continuous.

The bids of bidder $i$ are contained in $\hat{R}_{i}=[\hat{b}_{i}(0),\hat{b}_{i}(1)]$.
We now define a new mechanism with $\check{B}=[0,1]$, $\check{\phi}(w)=w$
and $\check{W}_{i},\check{L}_{i}:[0,1]^{2}\rightarrow\mathbb{R}$
given by: 
\begin{align*}
\check{W}_{i}(\check{b}_{i},\check{b}_{j}) & =\tilde{W}_{i}\left(\hat{b}_{i}(0)+\check{b}_{i}|\hat{R}_{i}|,\,\tilde{\phi}_{j}\left(\hat{b}_{i}(0)+\check{b}_{j}|\hat{R}_{i}|\right)\right),\\
\check{L}_{i}(b_{i},b_{j}) & =\tilde{L}_{i}\left(\hat{b}_{i}(0)+\check{b}_{i}|\hat{R}_{i}|,\,\tilde{\phi}_{j}\left(\hat{b}_{i}(0)+\check{b}_{j}|\hat{R}_{i}|\right)\right).
\end{align*}
The new mechanism has an equilibrium given by the functions $\check{b}_{i}^{s}(w)=(\hat{b}_{i}(w)-\hat{b}_{i}(0))/|\hat{R}_{i}|$
and $\check{b}_{j}^{s}(w)=(\hat{b}_{i}(w)-b_{i}(p))/|\hat{R}_{i}|$.
This equilibrium allocates to the bidder with the highest valuation
since $\check{b}_{i}^{s}(w)>\check{b}_{j}^{s}(w)$ if and only if
$\hat{b}_{i}(w)>\phi_{i}(\hat{b}_{i}(w))$ and the original mechanism
was efficient. This implies that all bidding functions are the same:
$\check{b}_{i}^{s}(w)=\check{b}_{j}^{s}(w)=:\check{b}(w)$ for $s\in\{r,m\}$.
Moreover $\check{W}_{i}$ and $\check{L}_{i}$ are $\mathcal{C}\text{\textonesuperior}$
since $\tilde{\phi}_{j}$ is continuously differentiable.

\textbf{Step 3: }The bidding function $\check{b}(w)$ is strictly
increasing and can therefore be decomposed as $\check{b}(w)=\check{b}^{C}(w)+\check{b}^{J}(w)$,
where $\check{b}^{C}(w)$ is continuous and $\check{b}^{J}(w)$ is
constant except for a countable number of jump-discontinuities. We
can modify the definition of $\check{M}$ and obtain a new smooth
auction-like mechanism $M$ with a symmetric equilibrium in which
$b(w)=\check{b}^{C}(w)/\left(\check{b}^{C}(1)-\check{b}^{C}(0)\right)$.

The function $b(w_{i})$ specifies an equilibrium in the mechanism
given by: 
\begin{align*}
W_{i}(b_{1},b_{2}) & =\check{W}_{i}(\check{b}((\check{b}^{C})^{-1}(b_{1}(b^{C}(1)-b^{C}(0)))),\check{b}((\check{b}^{C})^{-1}(b_{2}(b^{C}(1)-b^{C}(0))))),\\
L_{i}(b_{1},b_{2}) & =\check{L}_{i}(\check{b}((\check{b}^{C})^{-1}(b_{1}(b^{C}(1)-b^{C}(0)))),\check{b}((\check{b}^{C})^{-1}(b_{2}(b^{C}(1)-b^{C}(0))))).
\end{align*}

Next, we show that $W$ and $L$ are continuously differentiable.
In the mechanism defined in step 2, a rational bidder chooses $b_{i}$
to maximize 
\[
\int_{0}^{\check{b}^{-1}(b_{i})}\left(w_{i}-\check{W}_{i}(b_{i},\check{b}(w_{j}))\right)dF(w_{j}|w_{i},x_{i})-\int_{\check{b}^{-1}(b_{i})}^{1}\check{L}_{i}(b_{i},\check{b}(w_{j}))dF(w_{j}|w_{i},x_{i}),
\]
where $F(w'_{j}|w,x_{i})$ is the probability that $w_{j}\le w'_{j}$,
conditional on bidder $i$'s type $(w_{i},x_{i})$.

Consider a rational bidder with type $w_{i}=\hat{w}+\varepsilon$,
where $\hat{w}$ is a discontinuity in the equilibrium bidding function
$\check{b}$ of original mechanism. Placing a bid $b'\in[\check{b}(\hat{w}),\check{b}(\hat{w}_{+}))$
instead of $\check{b}(w_{i})$ must not be profitable: 
\begin{align*}
 & \int_{0}^{\check{b}^{-1}(\check{b}(w_{i}))}\left(w_{i}-\check{W}_{i}(\check{b}(w_{i}),\check{b}(w_{j}))\right)dF(w_{j}|w_{i},x_{i})-\int_{\check{b}^{-1}(\check{b}(w_{i}))}^{1}\check{L}_{i}(\check{b}(w_{i}),\check{b}(w_{j}))dF(w_{j}|w_{i},x_{i})\\
\ge & \int_{0}^{\check{b}^{-1}(b')}\left(w_{i}-\check{W}_{i}(b',\check{b}(w_{j}))\right)dF(w_{j}|w_{i},x_{i})-\int_{\check{b}^{-1}(b')}^{1}\check{L}_{i}(b',\check{b}(w_{j}))dF(w_{j}|w_{i},x_{i})
\end{align*}
This can be rewritten as 
\begin{align*}
 & \int_{0}^{\hat{w}}\left(w_{i}-\check{W}_{i}(\check{b}(w_{i}),\check{b}(w_{j}))\right)dF(w_{j}|w_{i},x_{i})-\int_{\hat{w}}^{1}\check{L}_{i}(\check{b}(w_{i}),\check{b}(w_{j}))dF(w_{j}|w_{i},x_{i})\\
 & +\int_{\hat{w}}^{\hat{w}+\varepsilon}\left(w_{i}-\check{W}_{i}(\check{b}(w_{i}),\check{b}(w_{j}))\right)dF(w_{j}|w_{i},x_{i})+\int_{\hat{w}}^{\hat{w}+\varepsilon}\check{L}_{i}(\check{b}(w_{i}),\check{b}(w_{j}))dF(w_{j}|w_{i},x_{i})\\
\ge & \int_{0}^{\hat{w}}\left(w_{i}-\check{W}_{i}(b',\check{b}(w_{j}))\right)dF(w_{j}|w_{i},x_{i})-\int_{\hat{w}}^{1}\check{L}_{i}(b',\check{b}(w_{j}))dF(w_{j}|w_{i},x_{i})
\end{align*}
The second term in on the left-hand side vanishes as $\varepsilon\rightarrow0$
since $\check{W}_{i}$ and $\check{L}_{i}$ are bounded. Hence we
must have 
\begin{multline*}
\int_{0}^{\hat{w}}\left(\check{W}_{i}(b',\check{b}(w_{j}))-\check{W}_{i}(\check{b}(\hat{w}_{+}),\check{b}(w_{j}))\right)dF(w_{j}|\hat{w},x_{i})\\
+\int_{\hat{w}}^{1}\left(\check{L}_{i}(b',\check{b}(w_{j}))-\check{L}_{i}(\check{b}(\hat{w}_{+}),\check{b}(w_{j}))\right)dF(w_{j}|\hat{w},x_{i})\ge0
\end{multline*}
Since $b'<\check{b}(w_{i})$, and $\check{W}_{i}$ and $\check{L}_{i}$
are non-decreasing in the first argument, this implies that $\check{W}_{i}(b',\check{b}(w_{j}))=\check{W}_{i}(b_{i},\check{b}(w_{j}))$
and $\check{L}_{i}(b',\check{b}(w_{j}))=\check{L}_{i}(b_{i},\check{b}(w_{j}))$
all $b'\in[\check{b}(\hat{w}),\check{b}(\hat{w}_{+})]$ and almost
every $w_{j}$. By continuity of $\check{W}_{i}$ and $\check{L}_{i}$
the equalities must hold for all $w_{j}$. Hence since $\check{W}_{i}$
and $\check{L}_{i}$ are continuously differentiable, $\partial\check{W}_{i}(b',\check{b}(w_{j}))/\partial b_{i}=0$
and $\partial\check{L}_{i}(b',\check{b}(w_{j}))/\partial b_{i}=0$
for all $w_{j}$ and all $b'\in[\check{b}(\hat{w}),\check{b}(\hat{w}_{+})]$
and also $\partial\check{W}_{i}(b',\check{b}(w_{j}))/\partial b_{j}=\partial\check{W}_{i}(\check{b}(\hat{w}),\check{b}(w_{j}))/\partial b_{j}=\partial\check{W}_{i}(\check{b}_{+}(\hat{w}),\check{b}(w_{j}))/\partial b_{j}$
and $\partial\check{L}_{i}(b',\check{b}(w_{j}))/\partial b_{j}=\partial\check{L}_{i}(\check{b}(\hat{w}),\check{b}(w_{j}))/\partial b_{j}=\partial\check{L}_{i}(\check{b}(\hat{w}_{+}),\check{b}(w_{j}))/\partial b_{j}$
for all $b'\in[\check{b}(\hat{w}),\check{b}(\hat{w}_{+})]$ and all
$w_{j}$. Hence continuous differentiability is preserved by the elimination
of the gaps. 
\end{proof}

\subsection{Proof of Lemma \protect\ref{lem:implications_delta=00003D00003D0}}
\begin{proof}[Proof of Lemma \ref{lem:implications_delta=00003D00003D0}]
We first show that for all $i$ and $b_{i},b_{j}\in[0,1]$: $\partial W_{i}(b_{i},b_{j})/\partial b_{i}=0$
if $b_{j}<b_{i}$, and $\partial L_{i}(b_{i},b_{j})/\partial b_{i}=0$
if $b_{j}>b_{i}$.

Since $\delta_{i}(b)=0$ for all $b\in[0,1]$ we have that $\psi$
\[
\psi'(b)=\frac{\partial W_{i}(b,b)}{\partial b_{i}}+\frac{\partial W_{i}(b,b)}{\partial b_{j}}-\frac{\partial L_{i}(b,b)}{\partial b_{i}}-\frac{\partial L_{i}(b,b)}{\partial b_{j}}<\infty
\]
where finiteness follows from the assumption that $W_{i}$ and $L_{i}$
are continuously differentiable.

Now suppose that for some $w_{i}\in(0,1)$, $\int_{0}^{1}\frac{\partial P_{i}(b(w_{i}),b(w_{j}))}{\partial b_{i}}f(w_{j}|w_{i},x_{i})dw_{j}>0$.
The same derivation leading to \eqref{eq:liminf>=00003D00003D00003Dlimsup}
in the proof of Lemma \ref{lem:psi'>0}, together with $\delta_{i}(b(w_{i}))=0$
implies that 
\[
\liminf_{b\nearrow b(w_{i})}\frac{\psi(b(w_{i}))-\psi(b)}{b(w_{i})-b}=\infty.
\]
This contradicts $\psi'(b(w_{i}))<\infty$. Hence $\int_{0}^{1}\frac{\partial P_{i}(b(w_{i}),b(w_{j}))}{\partial b_{i}}f(w_{j}|w_{i},x_{i})dw_{j}=0$
for all $w_{i}\in[0,1]$. Since $\partial P(b_{i},b_{j})/\partial b_{i}\ge0$
by assumption, we therefore have $\partial P_{i}(b_{0},b(w_{j}))/\partial b_{i}=0$
for almost every $w_{j}$ and by continuity of $\partial W_{i}/\partial b_{i}$,
$\partial L_{i}/\partial b_{i}$ and $b$, this holds for all $w_{j}$.
Therefore $\partial_{b_{i}}W_{i}(b_{0},b)=0$ if $b<b_{0}$, and $\partial_{b_{i}}L_{i}(b_{0},b)=0$
if $b>b_{0}$.

To conclude the proof, note that individual rationality together with
$L_{i}(b_{i},b_{j})\ge0$ requires that $L_{i}(0,b_{j})=0$ for all
$b_{j}$.\footnote{Notice that this holds independent of our normalization that $v^{1}=0$,
since the lowest type never wins the object in a regular equilibrium.} Since $\partial L_{i}(b_{i},b_{j})/\partial b_{i}=0$ if $b_{j}>b_{i}$,
this implies that $L_{i}(b_{i},b_{j})=0$ for all $b_{i}\le b_{j}$.
Next, $\delta_{i}(b(w_{i}))=0$ implies $W_{i}(b_{i}(w),b_{i}(w))=w_{i}+L_{i}(b_{i}(w),b_{i}(w))=w_{i}$,
and since $\partial W_{i}(b_{i},b_{j})/\partial b_{i}=0$, $W_{i}(b_{i},b_{j}(w_{j}))=w_{j}$
whenever $b_{i}\ge b_{j}(w_{j})$. 
\end{proof}

\subsection{Proof of Lemma \protect\ref{lem:psi'>0}}
\begin{proof}[Proof of Lemma \ref{lem:psi'>0}]
Consider a rational bidder $i$ with types $(w_{0},x_{i})\in[0,1]^{K-1}$
and any sequence of valuations $w_{i}^{n}\nearrow w_{0}$. $w_{i}^{n}$
prefers to bid $b^{n}=b(w_{i}^{n})$ over bidding $b_{0}=b(w_{0})$.
Therefore 
\begin{align*}
 & \int_{0}^{\psi(b^{n})}\left(w_{i}^{n}-W_{i}(b^{n},b(w_{j}))\right)f(w_{j}|w_{i}^{n},x_{i})dw_{j}-\int_{\psi(b^{n})}^{1}L_{i}(b^{n},b(w_{j}))f(w_{j}|w_{i}^{n},x_{i})dw_{j}\\
\ge & \int_{0}^{\psi(b_{0})}\left(w_{i}^{n}-W_{i}(b_{0},b(w_{j}))\right)f(w_{j}|w_{i}^{n},x_{i})dw_{j}-\int_{\psi(b_{0})}^{1}L_{i}(b_{0},b(w_{j}))f(w_{j}|w_{i}^{n},x_{i})dw_{j}\\
\\
\iff\: & \frac{1}{b-b^{n}}\int_{0}^{\psi(b^{n})}\left(W_{i}(b_{0},b(w_{j}))-W_{i}(b^{n},b(w_{j}))\right)f(w_{j}|w_{i}^{n},x_{i})dw_{j}\\
 & +\frac{1}{b-b^{n}}\int_{\psi(b^{n})}^{1}\left(L_{i}(b_{0},b(w_{j}))-L_{i}(b^{n},b(w_{j}))\right)f(w_{j}|w_{i}^{n},x_{i})dw_{j}\\
\ge & \frac{1}{b-b^{n}}\int_{\psi(b^{n})}^{\psi(b_{0})}\left(w_{i}^{n}-W_{i}(b_{0},b(w_{j}))+L_{i}(b_{0},b(w_{j}))\right)f(w_{j}|w_{i}^{n},x_{i})dw_{j}
\end{align*}
Taking the $\limsup$ on both sides we get 
\[
\int_{0}^{1}\frac{\partial P_{i}(b_{0},b(w_{j}))}{\partial b_{i}}f(w_{j}|w_{0},x_{i})dw_{j}\ge\delta_{i}(b_{0})f(w_{0}|w_{0},x_{i})\limsup_{n\rightarrow\infty}\frac{\psi(b_{0})-\psi(b^{n})}{b_{0}-b^{n}}
\]
where $P_{i}(b_{i},b_{j})=W_{i}(b_{i},b_{j})+L_{i}(b_{i},b_{j})$.
Similarly, $w_{0}$ prefers to bid $b_{0}$ over $b^{n}$ for all
$n\in\mathbb{N}$: 
\begin{align*}
 & \int_{0}^{\psi(b_{0})}\left(w_{0}-W_{i}(b_{0},b(w_{j}))\right)f(w_{j}|w_{0},x_{i})dw_{j}-\int_{\psi(b_{0})}^{1}L_{i}(b_{0},b(w_{j}))f(w_{j}|w_{0},x_{i})dw_{j}\\
\ge & \int_{0}^{\psi(b^{n})}\left(w_{0}-W_{i}(b^{n},b(w_{j}))\right)f(w_{j}|w_{0},x_{i})dw_{j}-\int_{1}^{\psi(b^{n})}\left(L_{i}(b^{n},b(w_{j}))\right)f(w_{j}|w_{0},x_{i})dw_{j}\\
\\
\iff\: & \frac{1}{b_{0}-b^{n}}\int_{\psi(b^{n})}^{\psi(b_{0})}\left(w_{0}-W_{i}(b_{0},b(w_{j}))+L_{i}(b_{0},b(w_{j}))\right)f(w_{j}|w_{0},x_{i})dw_{j}\\
 & \ge\frac{1}{b_{0}-b^{n}}\int_{0}^{\psi(b^{n})}\left(W_{i}(b_{0},b(w_{j}))-W_{i}(b^{n},b(w_{j}))\right)f(w_{j}|w_{0},x_{i})dw_{j}\\
+ & \frac{1}{b-b^{n}}\int_{\psi(b^{n})}^{1}\left(L_{i}(b_{0},b(w_{j}))-L_{i}(b^{n},b(w_{j}))\right)f(w_{j}|w_{i}^{n},x_{i})dw_{j}
\end{align*}
Taking the $\liminf$ on both sides we get 
\[
\delta_{i}(b_{0})f(w_{0}|w_{0},x_{i})\liminf_{n\rightarrow\infty}\frac{\psi(b_{0})-\psi(b^{n})}{b_{0}-b^{n}}\ge\int_{0}^{1}\frac{\partial P_{i}(b_{0},b(w_{j}))}{\partial b_{i}}f(w_{j}|w_{0},x_{i})dw_{j}.
\]
Hence, for $\delta_{i}(b_{0})>0$ we have 
\begin{align}
\liminf_{n\rightarrow\infty}\frac{\psi(b_{0})-\psi(b^{n})}{b_{0}-b^{n}}\ge & \frac{\int_{0}^{1}\frac{\partial P_{i}(b_{0},b(w_{j}))}{\partial b_{i}}f(w_{j}|w_{0},x_{i})dw_{j}}{\delta_{i}(b_{0})f(w_{0}|w_{0},x_{i})}\ge\limsup_{n\rightarrow\infty}\frac{\psi(b_{0})-\psi(b^{n})}{b_{0}-b^{n}}\label{eq:liminf>=00003D00003D00003Dlimsup}
\end{align}
Notice that so far we have considered the case that $w^{n}<w_{0}$.
The same steps apply for the case that the sequence satisfies $w^{n}>w_{0}$.
Hence condition \eqref{eq:liminf>=00003D00003D00003Dlimsup} applies
for both cases. We have

\begin{equation}
\psi'(b_{0})=\psi_{-}'(b_{0})=\psi_{+}'(b_{0})=\frac{\int_{0}^{1}\frac{\partial P_{i}(b_{0},b(w_{j}))}{\partial b_{i}}f(w_{j}|\psi(b_{0}),x_{i})dw_{j}}{\delta_{i}(b_{0})f(\psi(b_{0})|\psi(b_{0}),x_{i})}.\label{eq:ODE_phi}
\end{equation}
Hence $\psi(b_{0})$ is differentiable at $b_{0}$. Since $\delta_{i}(b)$
is continuous, there exists $\varepsilon$ such that $\delta_{i}(b)>0$
for all $b\in B_{\varepsilon}(b_{0})$. Since the right-hand side
of \eqref{eq:ODE_phi} is continuous in $b_{0}$, $\psi$ is continuously
differentiable on $B_{\varepsilon}(b_{0})$. Since $\psi$ is strictly
increasing there must be $b'\in B_{\varepsilon}(b_{0})$ such that
$\psi'(b')>0$ and since $\psi'$ is continuous, there exist $\alpha<b'<\beta$
such that $(\alpha,\beta)\subset B_{\varepsilon}(b_{0})$ and $\psi$
is continuously differentiable with $\psi'(b)>0$ for $b\in(\alpha,\beta)$. 
\end{proof}

\subsection{Proof of Lemma \protect\ref{lem:genericity_conditional_density}}

The proof follows the same steps as the proof of Theorem 2.4 in GH17,
except that instead of considering continuous mappings from $T_{i}$
to the space of all measures on $T_{-i}$, $\mathcal{M}(T_{-i})$,
we consider continuous mappings from $[0,1]$ to the space of all
absolutely continuous measures on $X=[0,1]^{K-2}$ with strictly positive
and continuous density, which we denoted by $\mathcal{M}_{+}^{d}(X)$.

Restricting attention to $\mathcal{M}_{+}^{d}(X)$ instead of the
space of all measures $\mathcal{M}(X)$, requires a straightforward
modification of the constructions of the functions $\boldsymbol{g}$
and the measures $\beta_{1},\ldots,\beta_{K}$ in footnote 20 of GH17.
First we take the functions $g^{k}$ to be functions $g^{k}:X\rightarrow[0,2]$
with $g^{k}(x^{k})=2$ and $g^{k}(x)=0$ for $x\notin B^{k}$. This
allows us to construct perturbations of the measures $\beta_{k}^{0}$
which need to be elements $\mathcal{M}_{+}^{d}(X)$ for our purposes,
by setting $\beta_{k}=(1-\varepsilon)\beta_{k}^{0}+\varepsilon\tilde{\beta}_{k}$
where the measure $\tilde{\beta}_{k}$ has a density $\tilde{f}_{k}$
that satisfies $\tilde{f}_{k}(x)$ for $x\notin B^{k}$ and $\int_{X}g^{k}(x)\tilde{f}_{k}(x)dx=1$.
Then, with $\varepsilon\neq-z/(1-z)$ for all negative eigenvalues
of the matrix $\left(\int_{X}g^{k}(x)\beta_{\ell}^{0}(dx)\right)_{k,\ell}$,
the vectors $\int_{X}\boldsymbol{g}(x)\beta_{k}(dx)$ for $k=1,\ldots,K$
are linearly independent. The remaining steps in the proof are virtually
unchanged.

\subsection{Proof of Lemma \protect\ref{lem:genericity_prior}}

The proof follows Theorem 2.7 in GH17 and uses results from Section
5.4 in \citet[henceforth GH14]{Gizatulina2014}.

First note that for elements of $\mathcal{M}_{+}^{d}([0,1]^{2K})$,
marginal and conditional densities are defined in the usual way. Moreover,
for each $w_{i}$, the function that maps $w_{j}$ to the conditional
probability measure on $X$ that is given by the density $f(x_{i}|w_{i},w_{j})$,
is an element of $\mathcal{C}([0,1],\mathcal{M}_{+}^{d}(X))$ (see
GH14).

Analog to the proof of Theorem 2.7 in GH17, we let $\mathcal{F}_{w_{i}}^{i}\subset\mathcal{M}_{+}^{d}([0,1]^{2K})$
be the set of priors such that the function $w_{j}\mapsto f(\cdot|w_{i},w_{j})$
is an element of $\mathcal{E}(w_{i})$. The key step is to show that
the residualness of $\mathcal{E}(w_{i})$ in $\mathcal{C}([0,1],\mathcal{M}_{+}^{d}(X))$
implies the residualness of $\mathcal{F}=\bigcap_{i\in\{1,2\},w_{i}\in\mathcal{W}_{i}}\mathcal{F}_{w_{i}}^{i}$
in $\mathcal{M}_{+}^{d}([0,1]^{2K})$. For each $i\in\{1,2\}$ and
$w_{i}\in(0,1)$, let $\psi_{i,w_{i}}:\mathcal{M}_{+}^{d}([0,1]^{2K})\rightarrow\mathcal{M}_{+}^{d}([0,1])\times\mathcal{C}([0,1],\mathcal{M}_{+}^{d}(X))$
be the mapping that maps the prior to the conditional distribution
$f(w_{j}|w_{i})$ and the function $w_{j}\mapsto f(x_{i}|w_{i},w_{j})$.
As shown in the proof of Lemma 5.9 in GH14, the maps $\psi_{i,w_{i}}$
are continuous and open if $\mathcal{M}_{+}^{d}([0,1]^{2K})$ is endowed
with the uniform topology for density functions. As in the proof of
Theorem 2.7 in GH17, this implies that $\mathcal{F}_{w_{i}}^{i}$
is as residual subset of $\mathcal{M}_{+}^{d}([0,1]^{2K})$, that
is it contains a countable intersection $\bigcap_{n\in\mathbb{N}}H_{n}(i,w_{i})$
of open and dense sets $H_{n}(i,w_{i})\subset\mathcal{M}_{+}^{d}([0,1]^{2K})$.
Clearly, $H=\bigcap_{i\in\{1,2\}}\bigcap_{w_{i}\in\mathcal{W}_{i}}\bigcap_{n(i,w_{i})\in\mathbb{N}}H_{n(i,w_{i})}(i,w_{i})$
is a subset of $\mathcal{F}$. By a diagonal argument, $H$ is a countable
intersection of open and dense subsets of $\mathcal{M}_{+}^{d}([0,1]^{2K})$
and hence $\mathcal{F}$ is residual.

\subsection{Proof of Lemma \protect\ref{lem:m-type_dont_bid_truthfully}}
\begin{proof}[Proof of Lemma \ref{lem:m-type_dont_bid_truthfully}]
We have shown this for $|V|=2$ in Proposition \ref{prop:SPA_is_inefficient}.
For $|V|\ge3$, we need to modify the proof. If $m$-types bid $b(w_{i})$,
we must have for all $\theta_{i}$ that 
\[
w_{i}=\mathbb{E}[v_{i}|\theta_{i}]\in\arg\max_{b}\left\{ \sum_{k=1}^{K}\theta_{i}^{k}\left(v_{i}^{k}H^{\text{SPA}}(b\mid v_{i}^{k})-\int_{0}^{b}zdH^{\text{SPA}}(z|v_{i}^{k})\right)\right\} .
\]
The first-order condition is 
\[
\sum_{k=1}^{|V|}\theta_{i}^{k}\left(v_{i}^{k}-w_{i}\right)H^{\text{SPA}\prime}(w_{i}\mid v_{i}^{k})=0
\]
Considering the type $\theta_{i}=(1-b,0,\ldots,0,b)$ for any $b\in(0,1)$,
we have $w_{i}=b$, and the first-order condition simplifies to 
\[
H^{\text{SPA}\prime}(b\mid v_{i}=1)-H^{\text{SPA}\prime}(b\mid v_{i}=0)=0
\]
We have 
\begin{align*}
H^{\text{SPA}}(b\mid v_{i}^{k}) & =\frac{\mathbb{P}_{f}\left[b_{j}\leq b,v_{i}=v_{i}^{k}\right]}{\mathbb{P}_{f}\left[v_{i}=v_{i}^{k}\right]}=\frac{\int_{\Theta_{i}}\mathbb{P}_{f}\left[w_{j}\le b\middle|\tilde{\theta}_{i}\right]\mathbb{P}_{f}\left[v_{i}=v_{i}^{k}\middle|\tilde{\theta}_{i}\right]f(\tilde{\theta}_{i})d\tilde{\theta}_{i}}{\mathbb{E}_{f}\left[\theta_{i}^{k}\right]}\\
 & =\frac{\int_{\Theta_{i}}F_{w_{j}}(b|\tilde{\theta}_{i})\tilde{\theta}_{i}^{k}f(\tilde{\theta}_{i})d\tilde{\theta}_{i}}{\mathbb{E}_{f}\left[\theta_{i}^{k}\right]}\\
H^{\text{SPA}\prime}(b\mid v_{i}^{k}) & =\frac{\int_{\Theta_{i}}f_{w_{j}}(b|\tilde{\theta}_{i})\tilde{\theta}_{i}^{k}f(\tilde{\theta}_{i})d\tilde{\theta}_{i}}{\mathbb{E}_{f}\left[\theta_{i}^{k}\right]}
\end{align*}
Substituting this in the first-order condition, we get for all $b\in B$:
\begin{align*}
\frac{\int_{\Theta_{i}}f_{w_{j}}(b|\tilde{\theta}_{i})\tilde{\theta}_{i}^{K}f(\tilde{\theta}_{i})d\tilde{\theta}_{i}}{\mathbb{E}_{f}\left[\theta_{i}^{K}\right]}-\frac{\int_{\Theta_{i}}f_{w_{j}}(b|\tilde{\theta}_{i})\tilde{\theta}_{i}^{1}f(\tilde{\theta}_{i})d\tilde{\theta}_{i}}{\mathbb{E}_{f}\left[\theta_{i}^{1}\right]} & =0\\
\iff\qquad\int_{\Theta_{i}}\left[\frac{\tilde{\theta}_{i}^{K}}{\mathbb{E}_{f}\left[\theta_{i}^{K}\right]}-\frac{\tilde{\theta}_{i}^{1}}{\mathbb{E}_{f}\left[\theta_{i}^{1}\right]}\right]f_{w_{i}}(\tilde{\theta}_{i}|w_{j}=b)f_{w_{j}}(b)d\tilde{\theta}_{i} & =0\\
\iff\qquad\frac{\mathbb{E}_{f}\left[\theta_{i}^{K}|w_{j}=b\right]}{\mathbb{E}_{f}\left[\theta_{i}^{K}\right]} & =\frac{\mathbb{E}_{f}\left[\theta_{i}^{1}|w_{j}=b\right]}{\mathbb{E}_{f}\left[\theta_{i}^{1}\right]}\\
\iff\qquad\mathbb{E}_{f}\left[\theta_{i}^{K}|w_{j}\le b\right] & =\frac{\mathbb{E}_{f}\left[\theta_{i}^{K}\right]}{\mathbb{E}_{f}\left[\theta_{i}^{1}\right]}\mathbb{E}_{f}\left[\theta_{i}^{1}|w_{j}\le b\right]
\end{align*}
For generic distributions, the last line is violated. 
\end{proof}
\newpage{}

 \bibliographystyle{economet}
\bibliography{JM}

\end{document}